\documentclass[11pt]{article}
\usepackage{bbm}
\usepackage{geometry}
\usepackage{amsfonts}
\usepackage{latexsym, amssymb, amsmath, amscd, amsthm, amsxtra}
\usepackage{mathtools}
\usepackage{enumerate}
\usepackage[all]{xy}
\usepackage{mathrsfs}
\usepackage{fancyhdr}
\usepackage{listings}
\usepackage{hyperref}
\usepackage{xcolor}
\hypersetup{colorlinks=true}
\def\E{\mathbb{E}}

\def\R{\mathbb{R}}

\def\11{\mathbbm{1}}

\def\tbf{\textbf }
\def\be{\begin{equation}}
\def\ee{\end{equation}}

\def\bR{\mathbb{R}}
\def\bE{\mathbb{E}}
\def\bP{\mathbb{P}}

\def\bN{\mathbb{N}}

\def\bZ{\mathbb{Z}}

\def\cA{\mathcal{A}}

\def\cP{\mathcal{P}}
\def\cW{\mathcal{W}}
\def\cS{\mathcal{S}}

\def\cF{\mathcal{F}}
\def\cG{\mathcal{G}}
\def\cL{\mathcal{L}}

\def\cN{\mathcal{N}}

\def\eps{\varepsilon}

\def\wt{\widetilde }
\def\ul{\underline    }
\def\tU{\text{U}}
\def\tV{\text{V}}
\def\tM{\text{M}}

\def\tU{\text{U}}
\def\tV{\text{V}}
\def\tW{\text{W}}
\def\tZ{\text{Z}}
\newcommand\dif{\mathop{}\!\mathrm{d}}

\newcommand{\Ex}[1]{\left\langle{#1}\right \rangle}

\newtheorem{thm}{Theorem}[section]

\newtheorem{proposition}[thm]{Proposition}

\newtheorem{lemma}[thm]{Lemma}

\theoremstyle{definition}

\numberwithin{equation}{section}
\usepackage{environ}

\begin{document}

\title{Operator Norm Bounds on the Correlation Matrix of the \\  SK Model at High Temperature}

\author{Christian Brennecke\thanks{Institute for Applied Mathematics, University of Bonn, Endenicher Allee 60, 53115 Bonn, Germany} 
\and Changji Xu\thanks{Center of Mathematical Sciences and Applications, Harvard University, Cambridge MA 02138,USA}
\and Horng-Tzer Yau\thanks{Department of Mathematics, Harvard University, One Oxford Street, Cambridge MA 02138, USA}
}

\maketitle

\begin{abstract}
We prove that the correlation matrix $ \tbf M= (\langle \sigma_i  \sigma_j\rangle-\langle \sigma_i  \rangle\langle\sigma_j\rangle)_{1\leq i,j\leq N} \in \bR^{N\times N}$ of the Sherrington-Kirkpatrick model has the property that for every $\epsilon>0$ there exists $K_\epsilon>0$, that is independent of $N$, such that
		\[ \bP\big(  \| \tbf M   \|_{\text{op}}  \leq  K_{\epsilon}\big) \geq 1- \epsilon \]
for $N$ large enough, for suitable interaction and external field parameters $(\beta,h)$ in the replica symmetric region. In other words, the operator norm of $\tbf M$ is of order one with high probability. Our results are in particular valid for all $ (\beta,h)\in (0,1)\times (0,\infty)  $ and thus complement recently obtained results in \cite{EAG,BSXY} that imply the operator norm boundedness of $\tbf M$ for all $\beta<1$ in the special case of a vanishing external field. 
\end{abstract}

%%%%%%%%%%%%%%%%%%%%%%%%%%%%%%%%%%%%%%%%%%%%
%%%%%%%%%%%%%%%%%%%%%%%%%%%%%%%%%%%%%%%%%%%%
%%%%%%%%%%%%%%%%%%%%%%%%%%%%%%%%%%%%%%%%%%%%

\section{Setup and Main Results}\label{sec:intro}

We consider systems of $N$ interacting spins $ \sigma_i\in \{-1,1\}, i\in [N]=\{1,\dots,N\}$, described by the Sherrington-Kirkpatrick \cite{SK} Hamiltonian $H_N:\{-1,1\}^N\to\mathbb{R}$ which is defined by
		\begin{equation} \label{eq:HN}
		H_N(\sigma) = \beta \sum_{1\leq i < j\leq N} g_{ij} \sigma_{i}\sigma_j  + h \sum_{i=1}^N \sigma_i= \frac{\beta}{2} (\sigma, \tbf{G}\sigma ) +  h\, (\tbf 1, \sigma).
		\end{equation}
The symmetric matrix $ \textbf{G} = (g_{ij})_{1\leq i,j\leq N}$ is a GOE matrix, that is, up to the symmetry constraint its entries are i.i.d.\ centered Gaussian random variables of variance $N^{-1}$ for $i\neq j$ and we set $g_{ii} =0$ for simplicity. The standard Euclidean inner product in $\bR^N$ and its induced norm are denoted by $ (\cdot,\cdot)$ and $\| \cdot\|$, respectively. We assume $\beta > 0$, $h>0$ and we assume the $\{g_{ij}\}$ to be realized in some probability space $(\Omega, \cF, \bP)$. The expectation with regards to $\bP$ is denoted by $\mathbb{E}(\cdot)$ and the $L^p(\Omega, \cF, \bP)$ norms are $ \| \cdot\|_{p} = (\bE\, |\cdot|^p)^{1/p}$. 

Based on a novel representation of the entries of the two point correlation matrix
		\[\tbf M=\tbf M_{\beta,h} = (m_{ij} )_{1\leq i,j\leq N} =  \big(\langle \sigma_i  \sigma_j \rangle- \langle \sigma_i \rangle\langle \sigma_j \rangle\big)_{1\leq i,j\leq N} \in \R^{N\times N} \]
as sums over weights of self-avoiding paths, which is motivated by the results of \cite{ABSY}, we recently proved in \cite{BSXY} for the special case $h=0$ that at high temperature, $ \tbf M$ is asymptotically close to a resolvent of $\tbf G$, in the sense that
		\be  \label{eq:resh0}  \lim_{N\to \infty} \big\| \tbf M_{\beta,h=0} - ( 1+\beta^2 - \beta \textbf{G})^{-1}  \big \|_{\text{op}}   =0 \ee
in probability. Here, $\| \tbf A\|_{\text{op}} = \sup_{\tbf u\in\bR^N: \|\tbf u\|=1}\| \tbf A \tbf u\| $ denotes the standard operator norm, for $\tbf A \in \bR^{N\times N}$, and $\langle \cdot \rangle =\langle \cdot \rangle_{\beta,h}$ denotes the Gibbs measure induced by $H_N$, that is 
		\[ \langle f \rangle = \frac{1}{Z_N}  \sum_{\sigma \in \{-1,1\}^N}  f(\sigma) \,e^{H_N(\sigma)} \hspace{0.5cm} \text{with} \hspace{0.5cm} Z_N =    \sum_{\sigma \in \{-1,1\}^N}  e^{H_N(\sigma)} \]
for $ f: \{-1,1\}^N \to \mathbb{R} $. We also occasionally abbreviate the covariance of two observables $f,g:\{-1,1\}^N\to\mathbb{R}$ by $ \langle f;g \rangle = \langle f g\rangle - \langle f\rangle \langle g\rangle$ so that $m_{ij} = \langle \sigma_i  ;\sigma_j \rangle$. By standard properties of GOE matrices, the validity of \eqref{eq:resh0} naturally suggests the (well-known) existence of a phase transition at $\beta =1$ and \cite{BSXY} verifies the validity of \eqref{eq:resh0} indeed in the replica symmetric region at $h=0$, that is, for all $\beta <1 $. It implies in particular that 
		\be \label{eq:normlm}\lim_{N\to\infty} \| \tbf M \|_{\text{op}} = \frac1{ (1-\beta)^{2}},    \ee
so that the operator norm of $ \tbf M$ is typically of order one as long as $\beta <1$. The boundedness of $\| \tbf M \|_{\text{op}}$ was also proved independently in \cite{EAG}, in the sense that $ \bE\,\| \tbf M \|_{\text{op}} = O(1)$.  

Denoting by $R :\{-1,1\}^N\times \{-1,1\}^N \to\bR $ the overlap so that $R (\sigma,\tau) = \frac1N (\sigma,\tau)$ and by $ \tbf m = (m_1,\ldots, m_N)   =  ( \langle \sigma_1\rangle, \ldots, \langle \sigma_N\rangle) \in (-1,1)^N$ the magnetization, the close relation between replica symmetry and $ \tbf M$ is perhaps most easily illustrated by observing that
		\be \label{eq:OLvsM} \frac1{N^2} \|\tbf M\|_\text{op}^2 \leq  \big\langle  \big(   R - \langle R\rangle \big)^2 \big\rangle =   \frac1{N^2} \text{tr}\,\tbf M^*\tbf M + \frac2{N^2}(\tbf m, \tbf M \,\tbf m)   \leq \frac1N \big(\|\tbf M\|_\text{op}^2 + 2 \|\tbf M\|_\text{op}\big).  \ee
Thus, if $ \langle R\rangle =1- \frac{\text{tr}\, \tbf M}N $ concentrates, suitable control on $\|\tbf M\|_\text{op}$ implies replica symmetry. A blow-up of $\|\tbf M\|_\text{op}$ (e.g.\ in a sense comparable to \eqref{eq:normlm} at $(\beta,h) =(1,0)$), on the other hand, may heuristically suggest that $ R $ does not self-average which prohibits replica symmetry.

The fact that the operator norm of $\tbf M $ detects the RS-RSB phase transition temperature correctly in the case of a vanishing external field $h=0$ (for which $ \langle R\rangle_{\beta,h=0}=0$), by \eqref{eq:resh0}, \eqref{eq:normlm} and  \eqref{eq:OLvsM}, naturally suggests to study the norm boundedness of $\tbf M$ and its relation to the RS phase for $h\neq 0$. Although a complete proof is still lacking, it is generally expected that the SK model is replica symmetric for all $(\beta,h)$ satisfying the AT \cite{AT} condition\footnote{Expectations over effective random variables, like for instance $Z$ in Eq.\ \eqref{eq:AT} and Eq.\ \eqref{eq:defq}, that are independent of the disorder $\{g_{ij}\}$ in \eqref{eq:HN}, are denoted by $E(\cdot)$, which is to be distinguished from expectations $\bE(\cdot)$ over the disorder in \eqref{eq:HN}. } 
		\be \label{eq:AT}\beta^2 E \, \text{sech}^{4}(h + \beta\sqrt{q} Z) <1, \ee
where here and in the following $ q =q_{\beta,h} $ denotes the unique solution to
		\be \label{eq:defq} q = E\tanh^2(h +\beta\sqrt{q} Z) \ee
and where $Z\sim \cN(0,1)$ denotes a standard Gaussian independent of the disorder $\{g_{ij}\}$.	
		
Proving the norm boundedness of $\tbf M$, or an analogue of \eqref{eq:resh0}, for $h\neq 0$ under \eqref{eq:AT} is a challenging problem. In this work instead, we impose further assumptions on $(\beta,h)$ which ensure exponential concentration of the overlap, based on results obtained by Talagrand in \cite{talagrand2010mean-vol1,talagrand2010mean-vol2}. More precisely, we assume in the sequel the AT condition \eqref{eq:AT} together with 
		\be  \label{eq:AT+} (\partial_m \Phi )(m,q') \big|_{m = 1} < 0 \text{ for all } q <q' \leq 1,  \ee 
where for $0\leq m \leq 1$ and $q \leq q'\leq 1$ we set
	\[\Phi(m,q') = \log 2 + \frac{\beta^2}{4}(1-q')^2 -  \frac{\beta^2}{4}m ({q'}^2 - q^2) + \frac 1 m E \log E_{Z'} \cosh^m ( h + \beta \sqrt{q}Z + \beta \sqrt{q'-q}Z')\,.\]
Here, $ (Z, Z')\sim\cN(0,\text{id}_{\bR^2})$ and $ E_{Z'}(\cdot)$ denotes the expectation over $Z'$. Conditions \eqref{eq:AT} \& \eqref{eq:AT+} ensure that we consider parameters $(\beta,h)$ inside a suitable interior region of the replica symmetric phase of the model (see \cite[Chapter 13]{talagrand2010mean-vol2} for more details on this). They imply in particular locally uniform in $\beta$ exponential concentration of the overlap, which in turn enables an efficient computation of many observables. 

While we recall the results from \cite{talagrand2010mean-vol1, talagrand2010mean-vol2} that are based on \eqref{eq:AT} \& \eqref{eq:AT+} and that are used in the sequel in Section \ref{sec:kpt-est}, let us already point out that \eqref{eq:AT} \& \eqref{eq:AT+} are satisfied for every external field strength $h>0$ whenever $\beta <1$: indeed following \cite{talagrand2010mean-vol2} one has
		\[\begin{split}
		( \partial_m \Phi )(m,q') \big|_{m = 1} (q) &=0, \hspace{0.5cm} \partial_{q'} \Big( ( \partial_m \Phi )(m,q') \big|_{m = 1} \Big)(q)  = 0, \\
		 \partial_{q'}^2 \Big( ( \partial_m \Phi )(m,q') \big|_{m = 1}\Big) & =  \frac{\beta^2}{2}\Big(\beta^2 E \frac{E_{Z'}\,\text{sech}^{3}( h + \beta\sqrt q Z+ \beta \sqrt{q'-q}Z') }{E_{Z'}  \cosh (h + \beta\sqrt q Z+ \beta \sqrt{q'-q}Z')} -1  \Big)
		\end{split} \]
and simple monotonicity arguments as in \cite[Section 3]{chen2021almeida} imply
		\[\begin{split}
		 \beta^2 \frac{  E_{Z'}\,\text{sech}^{3}( h + \beta\sqrt q Z+ \beta \sqrt{q'-q}Z') }{E_{Z'}  \cosh (h + \beta\sqrt q Z+ \beta \sqrt{q'-q}Z')} &   \leq  \beta^2 \,E_{Z'}\,  \text{sech}^{4}( h + \beta\sqrt q Z+ \beta \sqrt{q'-q}Z').
		 \end{split}\]
Thus, $ ( \partial_m \Phi )(m,q') \big|_{m = 1} < 0 $ for all $ q < q'  \leq 1$ whenever $\beta <1$. 
	
Assuming from now on \eqref{eq:AT} \& \eqref{eq:AT+}, our strategy to prove the norm boundedness of $\tbf M $ combines recent ideas from \cite{Bolt1, Bolt2, ABSY, CT21, celentano2022sudakov, EAG, BSXY}. While conceptually similar to the main strategy in \cite{BSXY}, notice that for $h\neq 0$, there is no analogue of the exact graphical representation of the entries of $\tbf M$ for $h=0$ available. For this reason, the implementation of our strategy is quite different from \cite{BSXY} and it provides several results of independent interest: in the first step of our analysis, we compute $\tbf M$ up to an error which has vanishing Frobenius norm in the limit $N\to \infty$. To state this result precisely, we denote from now on by $ \langle \cdot \rangle^{(i)}$ the Gibbs measure obtained from $\langle \cdot\rangle$ by removing the spin $\sigma_i$. 
		
\begin{proposition} \label{prop:hTAP}
	Assume that $ (\beta,h)$ satisfy \eqref{eq:AT} \emph{\&} \eqref{eq:AT+}. Then, there exists a constant $C=C_{\beta,h} >0$ such that for all $1\leq i\neq j\leq N$, we have that
		\[ \label{eq:hTAP} \E \Big ( m_{ij} - \beta \big(1 - m_i^2\big)\sum_{k =1}^N g_{ik} m_{kj}^{(i)} \Big)^2 \leq CN^{-5/2}. \]
\end{proposition}

\noindent \textbf{Remark 1.1.} \emph{Eq.\ \eqref{eq:hTAP} shows that $ m_{ij} \approx \beta\big(1 - m_i^2\big)\sum_{k =1}^N g_{ik} m_{kj}^{(i)} $. Iterating the r.h.s., one obtains a representation of $m_{ij}$ as a sum over weights $w(\gamma)$ of self-avoiding paths $\gamma$ from $i$ to $j$: each edge $ e = \{i_1,i_2\} \in \gamma$ contributes a factor $\beta g_{i_1i_2}$ to $w(\gamma)$ and each vertex a factor $\big(1-(m_i^{(S)})^2\big)$, for suitable $S\subset [N]$. For more details on a heuristic derivation of the path representation, see \cite{ABSY,BvDS}. Since we do not use the path representation to control $\| \emph{\tbf M}\|_{\emph{op}}$, we omit further details on its derivation, but mention that it generalizes the path representation of $ \tbf M_{\beta, h=0}$ obtained in \cite{BSXY} in the range $\beta <1$. }

\vspace{0.1cm}
\noindent \textbf{Remark 1.2.} \emph{A similar bound as in Prop.\ \ref{prop:hTAP} was proved in \cite[Theorem 1.1]{ABSY}. Compared to the latter, Prop.\ \ref{prop:hTAP} substantially improves the decay rate $ N^{-1-\epsilon}$ from \cite{ABSY} to a decay rate of order $N^{-5/2}$. The novel key ingredient that we use here consists of moment bounds on a suitable class of multi-point correlation functions that are naturally linked to strong overlap control. For the details, see in particular Lemmas \ref{corrbound}, \ref{lm:delta-m-estimate} and \ref{g-delta-est} below. }
\vspace{0.2cm}

As discussed in detail in \cite{ABSY} (for small $\beta$), after expressing the $ m_{kj}^{(i)}$ through the original two point functions $m_{kj}$, \eqref{eq:hTAP} suggests the validity of the resolvent type equation  
		\be\label{eq:resh} \tbf M \approx \frac1{\tbf D + \beta^2(1-q) - 2\beta^2 N^{-1}\tbf m\otimes \tbf m  -\beta \tbf G  }, \ee
where $\tbf D \in\bR^{N\times N}$ is the diagonal matrix with entries 
		\be \label{eq:Dij} (\tbf D)_{ij} = m^{-1}_{ii} \delta_{ij}= \frac1{1-m_i^2}\delta_{ij}.\ee 
Motivated by this approximation, the second step of our analysis proves the following. 
\begin{proposition} \label{prop:M-H} Assume that $ (\beta,h)$ satisfy \eqref{eq:AT} \emph{\&} \eqref{eq:AT+}, let $q$ denote the solution of \eqref{eq:defq} and let $\tbf D \in\bR^{N\times N}$ be defined as in \eqref{eq:Dij}. Then we have that
		\[  \|(\emph{\tbf D} + \beta^2(1-q) -  \beta \emph{\tbf G} )\emph{\tbf M}  - \emph{id}_{\bR^{N }}  \|_{\emph{op}} \leq o(1) \| \emph{\tbf M}\|_{\emph{op}} + O(1)  \]
for an error $o(1)\to 0 $ as $N\to \infty$ in probability and where the error $O(1) \geq 0$ is of order one with high probability: there exists a constant $C=C_{\beta,h}>0$ such that
		\[ \bP \big( O(1) \leq K \big) \geq 1 - CK^{-2} \]
for every $K>0$. 
\end{proposition}
\noindent \textbf{Remark 1.3.} \emph{In view of the resolvent heuristics \eqref{eq:resh}, it seems natural to expect that  
		\be \label{eq:resconv}\lim_{N\to\infty} \big\|(\emph{\tbf D} + \beta^2(1-q) - 2\beta^2 N^{-1}\tbf m\otimes \tbf m-  \beta \emph{\tbf G} )\emph{\tbf M}  - \emph{id}_{\bR^{N }} \big \|_{\emph{op}}  = 0, \ee
with high probability. Since the methods presented in this paper do not allow for a simple proof of \eqref{eq:resconv} and since the norm boundedness of $\tbf M$ is all that matters in connection with replica symmetry (cf.\ \eqref{eq:OLvsM}), we leave a proof of \eqref{eq:resconv} open for possible future work.}
\vspace{0.2cm}

Finally, in the last step of our analysis we prove the norm boundedness of $ \tbf M $ by deriving a uniform lower bound on the matrix $\tbf D + \beta^2(1-q) -  \beta \tbf G $ appearing in Prop.\ \ref{prop:M-H}. Compared to the special case $h=0$, in which case the resolvent approximation of $\tbf M$ only involves the GOE interaction $\tbf G$, let us point out that it is quite non-trivial to relate the infimum of the spectrum of $\tbf D + \beta^2(1-q) -  \beta \tbf G $ to the AT condition \eqref{eq:AT} due to the implicit dependence of $\tbf D$ on $\tbf G$. Our bounds are summarized in the next result which is a direct consequence of combining recent ideas from \cite{Bolt1, Bolt2, CT21, celentano2022sudakov}. In its statement, $ (\tbf m^{(k)})_{k\geq 1}$ denotes Bolthausen's iterative TAP solution  \cite{Bolt1,Bolt2}, whose construction is recalled in detail in Section \ref{sec:AMP} below.

\begin{proposition} \label{prop:op-H}
Assume that $ (\beta,h)$ satisfy the AT condition \eqref{eq:AT} and let $q$ denote the solution of \eqref{eq:defq}. Then, there exists $c =c _{\beta,h}>0$, which is independent of $N\in\bN$, so that 
		\be\label{eq:TAPHess} \emph{\tbf D}_{\tbf m^{(k)}}  + \beta^2(1-q) - \frac{2\beta^2}{N} \tbf m^{(k)}\otimes \tbf m^{(k)} -\beta \emph{\tbf G}  \geq c \ee
for $k$ large enough, with probability tending to one as $N\to \infty$. Here, $\tbf m^{(k)}$ denotes Bolthausen's iterative TAP solution at step $k$ and $\tbf D_{\tbf m^{(k)}}\in \bR^{N\times N}$ has entries
		$$ (\tbf D_{\tbf m^{(k)}} )_{ij}=\frac1{1-\big(m_i^{(k)}\big)^2}\delta_{ij}. $$

Moreover, assuming $ (\beta,h)$ to satisfy both \eqref{eq:AT} \& \eqref{eq:AT+}, it follows that  
		\[    \emph{\tbf D}  + \beta^2(1-q)  -\beta \emph{\tbf G}\geq \emph{\tbf D}  + \beta^2(1-q) - \frac{2\beta^2}{N}\tbf m\otimes \tbf m -\beta \emph{\tbf G}  \geq c   \]
with probability tending to one as $N\to\infty$, where $\tbf D \in\bR^{N\times N}$ is defined as in \eqref{eq:Dij}. In particular, $\emph{\tbf D}  + \beta^2(1-q) -\beta \emph{\tbf G} $ is invertible with high probability and satisfies
		\[ \big\| \big(\emph{\tbf D}  + \beta^2(1-q) -\beta \emph{\tbf G} \big)^{-1}\big\|_{\emph{op}}\leq c^{-1} <\infty.\]
\end{proposition}
\noindent \textbf{Remark 1.4.} \emph{We point out that the proof of \eqref{eq:TAPHess} only requires $(\beta,h)$ to satisfy \eqref{eq:AT}. The additional assumption \eqref{eq:AT+} is used to show that $ \tbf m$ is close to $\tbf m^{(k)}$, applying the main result of \cite{CT21} (see Prop.\ \ref{prop:vHv} for the details). The proof of \eqref{eq:TAPHess} follows, up to a suitable modifications, from translating the arguments in \cite{celentano2022sudakov} to the present setting. }

\vspace{0.1cm}
\noindent \textbf{Remark 1.5.} \emph{Notice that the matrix on the l.h.s.\ in \eqref{eq:TAPHess} is, up to a negligible error\,\footnote{Under \eqref{eq:AT}, it holds true that $ q = N^{-1} \sum_{i=1}^N \big(m_i^{(k)}\big)^2  + o(1)$ for an error $o(1)$ which is such that $ \lim_{k\to \infty} \limsup_{N\to\infty} o(1)\to 0$, in probability.}, equal to the negative of the Hessian of the TAP free energy functional \cite{TAP} at $\tbf m^{(k)}$. In particular, the lower bound in \eqref{eq:TAPHess} resolves a recent open question from \cite{GSS}, which studies the limiting spectral distribution of the Hessian at $\tbf m^{(k)}$ and which provides an interesting characterization of Plefka's conditions \cite{P}. For further recent work related to the TAP approach, see also \cite{DS, BK, CPS1, CPS2, CFM, BY, BNSX, EMS, B, KSS, G, BH, BvDS, BS} and the references therein.}
\vspace{0.2cm}

Collecting Propositions \ref{prop:hTAP}, \ref{prop:M-H} and \ref{prop:op-H}, we arrive at the following main result.

\begin{thm} \label{thm:main}
Assume that $ (\beta,h)$ satisfy \eqref{eq:AT} \emph{\&} \eqref{eq:AT+}. Then, for every $\epsilon>0$ there exists a constant $K_{\epsilon}>0$, that is independent of $N$, such that  
		\[\bP\big(  \| \emph{\tbf M}  \, \|_{\emph{op}}  \leq  K_\epsilon\big) \geq 1- \epsilon \] 
for all $N$ large enough. 
\end{thm}
\begin{proof} 
On a set of probability tending to one as $N\to\infty$, Prop.\ \ref{prop:op-H} shows that
		\[ \big\|\big( \emph{\tbf D}  + \beta^2(1-q)) -\beta \emph{\tbf G}\big)^{-1}\big\|_{\text{op}} \leq {c_{\beta,h}}^{-1}\]
for some constant $ c_{\beta,h}>0$, that is independent of $N$. Using Prop.\ \ref{prop:M-H} and the norm bound $ \| \tbf A \tbf B\|_{\text{op}}\leq  \| \tbf A \|_{\text{op}} \|  \tbf B\|_{\text{op}}$ for all $\tbf A, \tbf B\in \bR^{N\times N}$, this implies 
		\[\begin{split}
		\| \tbf M\|_{\text{op}} &\leq   {c_{\beta,h}}^{-1} \big(1+ \|( {\tbf D} + \beta^2(1-q) -  \beta  {\tbf G} ) {\tbf M}  -  \text{id}_{\bR^{N }}  \|_{\text{op}} \big) \\
		& \leq {c_{\beta,h}}^{-1} \big(1+ O(1) \big) + {c_{\beta,h}}^{-1}  o(1) \| {\tbf M}\|_{\text{op}}
		\end{split}\]
so that 
		\[ \| \tbf M\|_{\text{op}}  \leq \frac{  c_{\beta,h}^{-1}+ O(1)}{1- c_{\beta,h}^{-1}o(1) }. \]
The claim now follows from the fact that $ \lim_{N\to\infty} o(1) = 0$ in the sense of probability and the bound $\bP \big( O(1) > K \big) \leq CK^{-2}$, for some $C>0$ and every $ K>0$. 
\end{proof}

The rest of this paper is devoted to the proofs of Prop.\ \ref{prop:hTAP}, Prop.\ \ref{prop:M-H} and Prop.\ \ref{prop:op-H}. Section \ref{sec:kpt-est} recalls several useful results from \cite{talagrand2010mean-vol1, talagrand2010mean-vol2} and uses them to derive decay estimates on suitable correlation functions. Applying tools from stochastic calculus as in \cite{ABSY}, this implies Prop.\ \ref{prop:hTAP} and Prop.\ \ref{prop:M-H}, which is explained in Section \ref{sec:proofpart1}. Finally, in Section \ref{sec:AMP} we use the results of \cite{Bolt1, Bolt2, CT21, celentano2022sudakov} to deduce Prop.\ \ref{prop:op-H}. 

%%%%%%%%%%%%%%%%%%%%%%%%%%%%%%%%%%%%%%%%%%%%
%%%%%%%%%%%%%%%%%%%%%%%%%%%%%%%%%%%%%%%%%%%%
%%%%%%%%%%%%%%%%%%%%%%%%%%%%%%%%%%%%%%%%%%%%

\section{Bounds on $(k,p)$-Point Correlation Functions}
\label{sec:kpt-est}

In this section, we derive suitable decay bounds on a certain class of correlation functions that occur naturally in the proofs of Propositions \ref{prop:hTAP} and \ref{prop:M-H}. Our bounds combine ideas from \cite{ABSY} and results by Talagrand \cite{talagrand2010mean-vol1, talagrand2010mean-vol2} that follow from the assumptions \eqref{eq:AT} \& \eqref{eq:AT+}. 

To efficiently employ the results from \cite{talagrand2010mean-vol1, talagrand2010mean-vol2} that are recalled below, it is convenient to work in a slightly more general setting compared to the previous Section \ref{sec:intro}. To be more precise, we consider in this section spin systems with Hamiltonian of the form
		\be\label{eq:HN+} H_N(\sigma) = \beta \sum_{1\leq i < j\leq N} g_{ij} \sigma_{i}\sigma_j  +  \sum_{i=1}^N h_i \sigma_i= \frac{\beta}{2} (\sigma, \tbf{G}\sigma ) +   (\tbf h, \sigma), \ee
where $\tbf h\in \bR^N$ is assumed to be a Gaussian random vector whose components are i.i.d. copies of a Gaussian random variable $h $ with mean and variance denoted by
		$$ h \sim\cN( \mu_h, \sigma^2_h).$$
We allow for the possibility that $\sigma_h^2 =0$ (so that all of the following results apply in particular to systems with deterministic field as in \eqref{eq:HN}) and we assume that $E \,h^2 > 0$ (in particular, the Gaussian field $\tbf h$ need not be centered). Under the latter assumption, recall from \cite{Gue} and \cite[Chapter 1]{talagrand2010mean-vol1} that there is a unique solution $$q=q_{\beta,\mu_h,\sigma_h^2}\in (0,1)$$ to \eqref{eq:defq}. Let us point out that in this slightly more general setting, the expectation $E(\cdot)$ in \eqref{eq:AT}, \eqref{eq:defq} and \eqref{eq:AT+} is taken over all of the independent random variables $Z, Z' $ and $h$. 
 
In the following, we use dynamical methods to control the correlation functions of SK models with Hamiltonians as in \eqref{eq:HN+}, but with possibly modified interaction coupling $\beta'$ and Gaussian field $\tbf h'$ (whose entries are i.i.d. copies of a modified Gaussian random variable $h' \sim \cN(\mu_{h'}, \sigma_{h'}^2)$). All considered parameters $(\beta', h')$ satisfy, however, conditions \eqref{eq:AT}  \& \eqref{eq:AT+} and this motivates us to define the parameter set
		\[ \cA_{RS^-}  = \big\{ (\beta,\mu_h,\sigma_h^2)\in (0,\infty)\times\bR \times [0,\infty):   (\beta,h) \text{ satisfy } \eqref{eq:AT} \, \&\, \eqref{eq:AT+}, \text{ for } h\sim \cN(\mu_h, \sigma_h^2)\big\}.  \]
Using the continuous dependence $ (\beta, \mu_h, \sigma_h^2)\mapsto q_{\beta,h} =q_{\beta, \mu_h, \sigma_h^2}$, it is straightforward to check that the set $ \cA_{RS^-} \subset (0,\infty)\times\bR \times [0,\infty)$ is open, which is frequently used below.
		
We denote by $ (\sigma^l)_{l\in\bZ}$ a sequence of i.i.d. samples from $\langle\cdot\rangle$, called replicas. Following standard conventions, the $l$-th replica $\sigma^l$ corresponds to the $l$-th coordinate in $ \prod_{j\in\bZ}  \{-1,1\}^N $ with product Gibbs measure $ \bigotimes_{j\in \bZ} \langle \cdot \rangle$. By slight abuse of notation, we write $ \langle\cdot\rangle$ also for expectations over functions of several replicas. Consider e.g. the overlap $ R: \{-1,1\}^N\times \{-1,1\}^N\to \bR$ which equals the normalized inner product of two replicas:
		\[ R(\sigma^l, \sigma^{l'}) =N^{-1}\sum_{i=1}^N \sigma_i^l\sigma_i^{l'} = N^{-1} (\sigma^l,\sigma^{l'}).  \]
Below, we abbreviate for simplicity $ R_{l,l'} =R(\sigma^l, \sigma^{l'})$ when computing Gibbs expectations of functions of the overlap of multiple replicas. Following the previous remark, we thus write for instance $\langle f(R_{1,2}, R_{3,4})\rangle = \big\langle  f\big( N^{-1}(\sigma^1,\sigma^2), N^{-1}(\sigma^3,\sigma^4)\big)\big\rangle^{\otimes 4}$ for every $f:\bR\to \bR$. 

\begin{thm}[{\cite[Theorem 13.7.1.]{talagrand2010mean-vol2}}]  \label{thm:conc}
Assume $(\beta,\mu_h,\sigma_h^2)\in \cA_{RS^-}$ and $h\sim \cN(\mu_h,\sigma_h^2)$. Let $ q_{\beta,h}$ denote the unique solution of \eqref{eq:defq}. Then, there is a constant $   K_{\beta,h}>0$ such that 
		\[ \bE\,  \big \langle  \exp \big( N (R_{1,2} - q_{\beta,h} \big)^2/K_{\beta,h}\big) \big\rangle_{\beta,h} \leq 2.  \]
Moreover, $ K_{\beta,h}=K_{\beta,\mu_h,\sigma_h^2}$ is locally bounded in $(\beta, \mu_h, \sigma_h^2)$: for $\delta=\delta_{\beta,h} >0$ small enough s.t. $  (\beta -\delta, \beta+\delta)\times (\mu_h-\delta,\mu_h+\delta)\times [(\sigma_h^2-\delta)_+,\sigma_h^2+\delta)\subset \cA_{RS^-}$, we find $K>0$ such that 
		\be \label{eq:conc+}  \sup_{ \substack{ \beta' \in (\beta -\delta, \beta+\delta), \\ \mu_h'\in (\mu_h-\delta,\mu_h+\delta),\\ \sigma_{h'}^2\in [(\sigma_h^2-\delta)_+,\sigma_h^2+\delta)} } \bE \,\big \langle  \exp \big( N (R_{1,2} - q_{\beta',h'} \big)^2/ K \big) \big\rangle_{\beta',h'} \leq 2,  \ee
where $ (\sigma_h^2-\delta)_+=\max(0,\sigma_h^2-\delta)$ and where $h'\sim \cN(\mu_{h'},\sigma_{h'}^2)$.
\end{thm}
\noindent \textbf{Remark 2.1.} \emph{That $K_{\beta,h}>0$ is locally bounded in $(\beta, \mu_h,\sigma^2_h)$ is not explicitly stated in \cite[Theorem 13.7.1.]{talagrand2010mean-vol2}, but it readily follows from the arguments in \cite[Sections 13.4-13.7]{talagrand2010mean-vol2}.}

\vspace{0.3cm}
Theorem \ref{thm:conc} implies useful decay bounds on a large class of correlation functions. Here, we rely on a direct consequence of \eqref{eq:conc+} explained in detail in \cite[Sections 1.8 \& 1.10]{talagrand2010mean-vol1}. For the precise statement, we recall from \cite{talagrand2010mean-vol1} for $ l,l'\in \bZ$ the notation 
		\be\label{eq:Tll} \begin{split} T_{l,l'} = N^{-1} ( \sigma^l - \tbf m, \sigma^{l'}-\tbf m), \hspace{0.5cm} T_{l} = N^{-1} ( \sigma^l - \tbf m, \tbf m), \hspace{0.5cm}T  = N^{-1} (   \tbf m,  \tbf m)\\ \end{split}\ee
for replicas $(\sigma^l)_{l \in \bZ}$ as well as
		\be\label{eq:ABC} \begin{split}
		\nu_1^2 &=  \frac{1-2q +q_4   }{1- \beta^2( 1-2q +q_4\,) }, \hspace{0.5cm} \nu_2^2 =  \frac{q -q_4   }{\big( 1- \beta^2( 1-2q +q_4\,)\big) \big( 1- \beta^2( 1-4q +3 q_4\,)\big) }, \\
		\nu_3^2 &=  \frac{q_4 -q^2  }{1- \beta^2( 1-4q +3 q_4\,)}  +  \frac{\beta^2(q_4 - q^2) A^2 }{1- \beta^2( 1-4q +3 q_4\,)} +  \frac{2\beta^2( 2q + q^2 -q_4)B^2  }{1- \beta^2( 1-4q +3 q_4\,)} ,
		\end{split}\ee
where $q$ is as in \eqref{eq:defq} and where $q_4 =  E \tanh^4(h + \beta \sqrt{q} Z)$. 

\begin{thm}[{\cite[Theorem 1.10.1]{talagrand2010mean-vol1}}]
\label{tal}
Assume $(\beta,\mu_h,\sigma_h^2)\in \cA_{RS^-}$ and $h\sim \cN(\mu_h,\sigma_h^2)$. Denote by $U_{l,l'},U_l$ and $U$ independent centered Gaussian random variables with variance $\nu_1^2, \nu_2^2$ and, respectively, $\nu_3^2$, as in \eqref{eq:ABC}. Let $ T_{l,l'}, T_l $ and $T$ be defined as in \eqref{eq:Tll}. Then, for fixed $n\in\bN_0$, $ k(l,l') \in \bN_0$ for $1\leq l<l'\leq n$, $ k(l) \in \bN_0$ for $1\leq l\leq n$, $k\in \bN_0$ as well as
		\[ m = \sum_{1\leq l < l' \leq n } k(l,l') +\sum_{1\leq l  \leq n } k(l)+k,  \]
we have that
		\[\begin{split}
		&\Big|  \bE\, \Big\langle  \prod_{1\leq l<l' \leq n} T_{l,l'}^{k(l,l')}  \prod_{1\leq l\leq n} T_{l}^{k(l)}  T^{k}   \Big\rangle  - N^{-\frac m2} E\!\! \prod_{1\leq l<l' \leq n} U_{l,l'}^{k(l,l')}   \prod_{1\leq l\leq n} U_{l}^{k(l)}  U^{k}\Big|\leq  C N^{-\frac{m+1}2} . 
		\end{split}\]
Moreover, the constant $ C=C_{\beta,h}$ is locally bounded in $ (\beta, \mu_h, \sigma_h^2)$. 
\end{thm}
\noindent \textbf{Remark 2.2.} \emph{We remark that \cite[Theorem 1.10.1]{talagrand2010mean-vol1} assumes $\beta <1/2$ in order to apply \cite[Eq. (1.88)]{talagrand2010mean-vol1} to obtain the error bounds in \cite[Eq. (1.296)]{talagrand2010mean-vol1}. The latter is the key result that implies \cite[Lemma 1.10.2 \& Corollary 1.10.3]{talagrand2010mean-vol1}, which in turn implies \cite[Theorem 1.10.1]{talagrand2010mean-vol1}. Instead of assuming $\beta<1/2$, we can equally assume $(\beta,\mu_h,\sigma_h^2)\in \cA_{RS^-}$ and $h\sim \cN(\mu_h,\sigma_h^2)$ to get \cite[Eq. (1.296)]{talagrand2010mean-vol1}, by Theorem \ref{thm:conc}. Given \cite[Eq. (1.296)]{talagrand2010mean-vol1}, Theorem \ref{tal} can then be proved exactly as in \cite[Section 1.10]{talagrand2010mean-vol1}. Moreover, the fact that the constant $C_{\beta,h}$ is locally bounded in $ (\beta, \mu_h, \sigma_h^2)$ is not explicitly explained in \cite[Section 1.10]{talagrand2010mean-vol1}, but it readily follows as a consequence of Theorem \ref{thm:conc} and the relevant tools from \cite[Sections 1.6, 1.8 \& 1.10]{talagrand2010mean-vol1}, in particular \cite[Eq. (1.151), (1.154), (1.215) \& (1.216)]{talagrand2010mean-vol1}.} 
\vspace{0.2cm}

We are now ready to apply the preceding results to derive suitable decay bounds on correlation functions. For $ k\in \bN$ and $i_1,\ldots, i_k \in [N]$, we define the $k$-point functions by 
	\be \label{eq:corrk}m_{i_1 i_2\ldots i_k} =  \partial_{h_{i_1}} \partial_{h_{i_2}}\ldots \partial_{h_{i_k}}  \log Z_N (\tbf h).  \ee
To derive our main results, it is useful to consider a certain class of $ (k,p)$-point correlation functions, which are defined as follows: fixing $k,p\in\bN_0$ and $ j_1,..,j_p \in [N]$, we define 
	\be \label{eq:corrkp}	m_{j_1\ldots j_p;k} = N^{-k}\sum_{i_1,\ldots,i_k = 1}^N  m_{i_1 i_1 i_2 i_2 \ldots i_ki_k j_1 j_2\ldots j_p}\, \ee
by which we mean for $k=0$ that $ m_{j_1\ldots j_p;0}=m_{j_1\ldots j_p}$. Notice that each of the indices $ i_1, \ldots, i_k $ over which we average occurs exactly twice.

\begin{lemma} \label{corrbound} Assume $(\beta,\mu_h,\sigma_h^2)\in \cA_{RS^-}$ and $h\sim \cN(\mu_h,\sigma_h^2)$. Set $ A_m = m/2$ if $m\leq 2$ and $A_m = m/2+1/2$ if $m\geq 3$.  
Then, we have for every $k \geq 0,p \geq 1$ with $2k+p \geq 2$ that
	\begin{align*} &N^{-p}\sum_{j_1,\ldots, j_p =1 }^N \E\, m_{j_1\ldots j_p;k}^2  \leq C_{\beta,h} N^{-A_{2k+p}}\end{align*}
for some $C_{\beta,h}>0$ that is locally bounded in $ (\beta, \mu_h, \sigma_h^2) $. 
\end{lemma}
\begin{proof}
For simplicity, we give a detailed proof for the cases $k=0$ and $k=1$; the remaining cases $k \geq 2$ are proved with analogous arguments. We use the representation	
		\begin{equation} \label{eq:k-pt-re} m_{j_1\ldots j_p} = \Big\langle \sigma^1_{j_1}\prod_{u = 1}^{p-1} \Big( \sum_{v = 1}^u(\sigma^v_{j_{u+1}} - \sigma^{u+1}_{j_{u+1}})\Big) \Big\rangle, \end{equation}
where we recall that $(\sigma^l)_{l\in\bN} $ denote i.i.d. replicas sampled from $\langle\cdot\rangle$. Eq.\ \eqref{eq:k-pt-re} follows from 
 		\begin{align*} \partial_{h_i}\Ex{f(\sigma^1,\ldots,\sigma^u)} &= \sum_{\sigma^1,\ldots,\sigma^u \in \{-1,1\}^N}f(\sigma^1,\ldots,\sigma^u)\,\partial_{h_i} \frac{e^{H_N(\sigma^1) +	\ldots+H_N(\sigma^u)}}{Z_N^u}\\
		& = \Big\langle f(\sigma^1,\ldots,\sigma^u)\Big(\sum_{v = 1}^u \sigma^v_i - u \sigma^{u+1}_i\Big)\Big\rangle\\
		& = \Big\langle f(\sigma^1,\ldots,\sigma^u) \sum_{v = 1}^u ( \sigma^v_i - \sigma^{u+1}_i) \Big\rangle \end{align*}
and induction. Observe that the previous identity implies in particular that  
		\begin{equation} \label{eq:zero2.4}
		\Big\langle \prod_{u = 1}^{p-1} \Big(\sum_{v = 1}^u(\sigma^v_{j_{u+1}} - \sigma^{u+1}_{j_{u+1}})\Big)\Big\rangle = \partial_{h_{j_2}} \ldots\partial_{h_{j_p}}\Ex{1} = 0\,. 
		\end{equation}
Now, consider first the case $k = 0$. Recalling that $ (\sigma^l)_{l\in\bZ}$ denotes a sequence of i.i.d. replicas sampled from $\langle\cdot\rangle$, Eq.\ \eqref{eq:k-pt-re} and \eqref{eq:zero2.4} imply that we can express $m_{j_1\ldots j_p}^2$ as
	\[\begin{split} m_{j_1\ldots j_p}^2 &= \Big\langle \big(\sigma^1_{j_1}\sigma^{-1}_{j_1}-N^{-1}\|\tbf m\|^2\big)\prod_{u = 1}^{p-1} \Big(\sum_{v = 1}^u(\sigma^v_{j_{u+1}} -  \sigma^{u+1}_{j_{u+1}})\Big)\Big(\sum_{v' = -1}^{-u}(\sigma^{v'}_{j_{u+1}} - \sigma^{-(u+1)}_{j_{u+1}})\Big)\Big\rangle\\
	%& = \Big\langle(\sigma^1_{j_1}\sigma^{-1}_{j_1}-N^{-1}\|\tbf m\|^2)\prod_{u = 1}^{p-1}  \sum_{v = 1}^u \sum_{v = -1}^{-u} \big( \sigma^v_{j_{u+1}}\sigma^{v'}_{j_{u+1}}  -  \sigma^{u+1}_{j_{u+1}} \sigma^{v'}_{j_{u+1}}   -  \sigma^v_{j_{u+1}} \sigma^{-(u+1)}_{j_{u+1}} + \sigma^{u+1}_{j_{u+1}}\sigma^{-(u+1)}_{j_{u+1}} \big)  \Big\rangle
	\end{split} \]
so that
	\[\begin{split}
		&N^{-p}\sum_{j_1,\ldots, j_p=1}^N m_{j_1\ldots j_p}^2\\
		& = \Big\langle \big(R_{1,-1}-N^{-1}\|\tbf m\|^2\big)\prod_{u = 1}^{p-1}  \sum_{v = 1}^u\sum_{v' = -1}^{-u} \big(R_{v,v'} - R_{v,-(u+1)} - R_{u+1,v'} + R_{u+1,-(u+1)} \big)\Big\rangle\\
		&  = \Big\langle(T_{1,-1}+T_1 + T_{-1} )\prod_{u = 1}^{p-1} \sum_{v = 1}^u\sum_{v' = -1}^{-u} \big(T_{v,v'} - T_{v',-(u+1)} - T_{u+1,v'} + T_{u+1,-(u+1)}\big)  \Big\rangle. 
	\end{split}\]
Here, we used in the second step that $R_{l,l'} = T_{l,l'} + T_l + T_{l'} + N^{-1}\|\tbf m\|^2$ and that 
		\[\begin{split}
		&R_{v,v'} - R_{v,-(u+1)} - R_{u+1,v'} + R_{u+1,-(u+1)}\\
		& = T_{v,v'} + T_v + T_{v'} + N^{-1}\|\tbf m\|^2  - T_{v,-(u+1)} - T_{v} - T_{-(u+1)} - N^{-1}\|\tbf m\|^2\\
		&\hspace{0.5cm} - T_{u+1,v'} - T_{u+1} - T_{v'} - N^{-1}\|\tbf m\|^2 + T_{u+1,-(u+1)} + T_{u+1} + T_{-(u+1)} + N^{-1}\|\tbf m\|^2\\
		& = T_{v,v'} - T_{v',-(u+1)} - T_{u+1,v'} + T_{u+1,-(u+1)}.
		\end{split}\]
Now, applying and using the same notation as in Theorem \ref{tal}, we get
	\begin{equation}
	\label{eq:3.3}
		 \Big| N^{-p}\sum_{j_1,\ldots, j_p=1}^N \bE \,m_{j_1\ldots j_p}^2 - N^{-\frac p2}E \prod_{u = 0}^{p-1} Z_u\Big| \leq  C N^{-\frac{p+1}2}, 
	\end{equation}
where $ (Z_u)_{0\leq u\leq p-1}$ is a Gaussian random vector whose entries are defined by 
		\[ Z_u = \begin{cases} U_{1,-1}+ U_1 + U_{-1}&: u=0, \\  \sum_{v = 1}^u\sum_{v' = -1}^{-u}\big(U_{vv'} - U_{v,-(u+1)} - U_{u+1,v'} + U_{u+1,-(u+1)}\big) &: u \geq 1. \end{cases}\]
This proves the lemma for the case $k=0$ and $p=2$. To obtain an additional decay factor $N^{-1/2}$ for the cases $p\geq 3$, notice that the entries of $ (Z_u)_{1\leq u\leq p-1}$ are independent:
		\[\begin{split}
		E Z_u Z_{u'} =& \,E\,  \sum_{v = 1}^u\sum_{v' = -1}^{-u} \sum_{w = 1}^{u'}\sum_{w' = -1}^{-u'} \big(U_{vv'} - U_{v,-(u+1)} - U_{u+1,v'} + U_{u+1,-(u+1)}\big) U_{ww'}\\
		=&\, \sum_{v = 1}^u\sum_{v' = -1}^{-u} \big( E\, U_{vv'}^2 -   E \, U_{v,-(u+1)}^2  -   E \, U_{u+1,v'}^2 +E\,  U_{u+1,-(u+1)}^2\big) = 0
		\end{split}\]
for all $1\leq u < u '\leq p-1$. Since $ (Z_u)_{0\leq u\leq p-1}$ is a Gaussian vector, we can then write
		\begin{equation*} Z_0 =  \sum_{u=1}^{p-1}x_u Z_u + Z', \end{equation*}
where $\tbf x = (x_1,\ldots, x_{p-1})\in \bR^{p-1}$ and where $Z'$ is a Gaussian random variable independent of the remaining entries $Z_1,\ldots,Z_{p-1}$. This implies for $p\geq 3$ that 
		\[E \prod_{u = 0}^{p-1} Z_u = \sum_{u=1}^{p-1} x_u  \big( E\, Z_u^2\big) \, E \prod_{\substack{ u'=1,\\ u'\neq u}}^{p-1}  Z_{u'} =0\]
and combining this with \eqref{eq:3.3}, we conclude that 
		\[ \Big| N^{-p}\sum_{j_1,\ldots, j_p=1}^N \bE \,m_{j_1\ldots j_p}^2 \Big| \leq  CN^{-\frac{p+1}2 }.\]
	
Next, consider the case $k = 1$. By \eqref{eq:k-pt-re}, we have that
		\begin{align*}
		m_{j_1\ldots j_p;1} =  \frac1N \sum_{i_1=1}^N m_{i_1i_1 j_1\ldots j_p} 
		& =  \frac 1 N \sum_{i_1=1}^N\Big\langle \sigma_{i_1}^1(\sigma_{i_1}^1 - \sigma_{i_1}^2)\prod_{u = 2}^{p+1} \Big(\sum_{v = 1}^u\big(\sigma^v_{j_{u-1}} - \sigma^{u+1}_{j_{u-1}}\big)\Big)\Big\rangle\\
		& =  \Big\langle (1 - R_{1,2})\prod_{u = 2}^{p+1} \Big(\sum_{v = 1}^u\big(\sigma^v_{j_{u-1}} - \sigma^{u+1}_{j_{u-1}}\big)\Big)\Big\rangle\\
		& = \Big\langle \big( N^{-1} \| \textbf m \|^2- R_{1,2}\big)\prod_{u = 2}^{p+1} \Big(\sum_{v = 1}^u\big(\sigma^v_{j_{u-1}} - \sigma^{u+1}_{j_{u-1}}\big)\Big)\Big\rangle,
		\end{align*}
where we used \eqref{eq:zero2.4} in the last step. Hence
	\[\begin{split}
		m_{j_1\ldots j_p;1}^2 &= \Big\langle \big(N^{-1} \| \textbf m \|^2-R_{12}\big)\big(N^{-1} \| \textbf m \|^2-R_{-1,-2}\big) \\
		&\hspace{1cm}\times \prod_{u = 2}^{p+1} \Big(\sum_{v = 1}^u(\sigma^v_{j_{u-1}} - \sigma^{u+1}_{j_{u-1}})\Big)\Big(\sum_{v' = -1}^{-u}(\sigma^{v'}_{j_{u-1}} -  \sigma^{-(u+1)}_{j_{u-1}})\Big)\Big\rangle
	\end{split}\]
and thus, arguing as in the first step, 
	\begin{align*}
		  N^{-p}\!\!\! \sum_{j_1,\ldots,j_p=1}^N \!\!\!\bE \, m_{j_1\ldots j_p;1}^2
		%&= \E \Ex{(q-R_{12})(q-R_{-1,-2})\prod_{u = 2}^{p+1} \Big(\sum_{v = 1}^u\sum_{v' = -1}^{-u}(R_{v,v'} - R_{v,-(u+1)} - R_{u+1,v'} + R_{u+1,-(u+1)}) \Big)}\\
		& = \bE \,\Big\langle (T_1 + T_2 + T_{12})(T_{-1} + T_{-2} + T_{-1,-2}) \\
		&\hspace{1cm}\times\prod_{u = 2}^{p+1} \Big(\sum_{v = 1}^u\sum_{v' = -1}^{-u}\big(T_{v,v'} - T_{v,-(u+1)} - T_{u+1,v'} + T_{u+1,-(u+1)}\big) \Big)\Big\rangle.
		% & = \E[(U_1 + U_2 + U_{12})(U_{-1} + U_{-2} + U_{-1,-2})X] + \mathcal O(N^{-(p+1)/2-1})
	\end{align*}
Applying and using the same notation as in Theorem \ref{tal}, we obtain that 
	\[\Big|  N^{-p}\!\!\! \sum_{j_1,\ldots,j_p=1}^N \!\!\!\bE \, m_{j_1\ldots j_p;1}^2 - E\, (U_1 + U_2 + U_{12})(U_{-1} + U_{-2} + U_{-1,-2}) Y_{1,p}\Big|\leq CN^{-\frac{p+3}2 }, \]
where we set 	
		$$Y_{1,p} =\prod_{u = 2}^{p+1} \Big(\sum_{v = 1}^u\sum_{v' = -1}^{-u}\big (U_{vv'} - U_{v,-(u+1)} - U_{u+1,v'} + U_{u+1,-(u+1)}\big)\Big).$$
Notice in particular that $ Y_{1,p}$ is independent of the random variables $ (U_{l,l'})_{l,l' \in \bZ: ll'>0 } $ and $ (U_{l})_{l\in\bZ}$ and that $E Y_{1,p} =0$, arguing as in the previous step. We therefore conclude that
		\[\Big|  N^{-p}\!\!\! \sum_{j_1,\ldots,j_p=1}^N \!\!\!\bE \, m_{j_1\ldots j_p;1}^2  \Big|\leq CN^{-\frac{p+3}2 } = C N^{-A_{2+p}}.\]

For the remaining cases $k\geq 2$, using the same arguments as above, we find that 
		\[ \Big|  N^{-p}\!\!\! \sum_{j_1,\ldots,j_p=1}^N \!\!\!\bE \, m_{j_1\ldots j_p;k}^2 - E\, X_{k,p} Y_{k,p} \Big|\leq CN^{-\frac{2k+p+1}2 } = N^{-A_{2k+p}},  \]
where 
		\[Y_{k,p} =\prod_{u = 2k}^{p+2k-1} \Big(\sum_{v = 1}^u\sum_{v' = -1}^{-u}\big (U_{vv'} - U_{v,-(u+1)} - U_{u+1,v'} + U_{u+1,-(u+1)}\big)\Big)\]
so that by independence $ E \,Y_{k,p}=0$, and where $X_{k,p}$ is a finite polynomial of the random variables $ (U_{l,l'})_{l,l' \in \bZ: ll'>0 } $ and $ (U_{l})_{l\in\bZ}$ so that $E [X_{k,p} Y_{k,p}] = 0$.
\end{proof}

We conclude this section with a few corollaries of Lemma \ref{corrbound} that provide useful bounds on the correlation functions, conditionally on a fixed spin $\sigma_i$. Recall that $ \textbf m\in\bR^N$ denotes the magnetization vector (i.e. $ m_i = \langle \sigma_i\rangle$) and that $ \textbf m^{(i)}$ denotes the magnetization of the system after removing $\sigma_i$. We write $\langle \cdot\rangle^{(i)}$ for the corresponding Gibbs measure so that
		\[ \textbf m^{(i)}  = \big(m_1^{(i)}, \ldots, m_{i-1}^{(i)}, m_{i+1}^{(i)} ,\ldots, m_N^{(i)}\big) = \big( \langle \sigma_1\rangle^{(i)}, \ldots,\langle \sigma_{i-1}\rangle^{(i)}, \langle\sigma_{i+1}\rangle^{(i)} , \ldots, \langle\sigma_N\rangle^{(i)}\big).  \] 
Note that $ \langle \cdot\rangle^{(i)}  $ is a Gibbs measure induced by the SK Hamiltonian as in \eqref{eq:HN}, but with $N-1$ particles and with coupling $ \beta' = (1-1/N)^{1/2}\beta $ (which is close to $\beta$ for $N\gg 1$). 

More generally, for fixed $i\in [N]$, let us identify $ (g_{ij})_{j:j\neq i}$ with $ (g_{ij}(1))_{j:j\neq i}$, where the components of $ \big( (g_{ij}(t) )_{j:j\neq i}\big)_{t\in [0,1]}$ denote $N-1$ independent Brownian motions of speed $1/N$, i.e. $ \bE  (g_{ij}(t))^2 = t/N$. Then, for $ \sigma_i\in\{-1,1\}$ and $t\in [0,1]$, we denote by $ \langle \cdot \rangle^{[i]}(\sigma_i, t)$ the Gibbs measure induced by the SK Hamiltonian 
		\[\begin{split}
		H_{N-1}^{[i]} (\sigma) &=  \sum_{\substack{ 1\leq u < v\leq N, \\ u,v\neq i} } \beta'  g'_{uv} \sigma_{u}\sigma_v  + \sum_{\substack{1\leq u\leq N: \\ u\neq i}}^N \big(h_u  + \beta g_{iu} (t) \sigma_i\big)  \sigma_u\\
		& = \sum_{\substack{ 1\leq u < v\leq N, \\ u,v\neq i} } \beta'  g'_{uv} \sigma_{u}\sigma_v  + \sum_{\substack{1\leq u\leq N: \\ u\neq i}}^N h_u'(\sigma_i,t)  \sigma_u
		\end{split} \]
for $ \sigma = (\sigma_j)_{j\in [N], j\neq i} \in \{-1,1\}^{N-1}$. Here, the interactions $ g'_{uv} = (1-1/N)^{-1/2} g_{uv}$ define a GOE matrix $\textbf{G}'\in  \bR^{N-1\times N-1}$ and the random field $\big( h_j'(\sigma_i,t)\big)_{j\in [N], j\neq i}$ consists of independent Gaussian copies of $h'\sim \cN(\mu_{h'}, \sigma_{h'}^2) =  \cN(\mu_h, \sigma_h^2+ \beta^2t /N)$. Observe that 
		\be \label{eq:bminb'} | \mu_h- \mu_{h'}| =0, \hspace{0.5cm} |\beta - \beta'| \to 0, \hspace{0.5cm} | \sigma^2_h- \sigma^2_{h'}| \to 0  \ee
as $N\to \infty$ and that $ \langle\cdot\rangle^{[i]} (\sigma_i,0) = \langle\cdot\rangle^{(i)}$. Moreover, observe that $ \langle\cdot\rangle^{[i]} (\cdot,1) $ corresponds to the Gibbs expectation $(\langle\cdot\rangle| \sigma_i )$ conditionally on $\sigma_i$. 

In the sequel, we denote by $ m_{j_1\ldots j_p}^{[i]} =m_{j_1\ldots j_p}^{[i]}(\sigma_i,t) $ the $k$-point and by $m_{j_1\ldots j_p;k}^{[i]} =m_{j_1\ldots j_p;k}^{[i]}(\sigma_i,t)$ the $(k,p)$-point correlation functions with regards to $\langle\cdot\rangle^{[i]}$, defined as in \eqref{eq:corrk} and \eqref{eq:corrkp}. Following \cite{ABSY}, we finally define for observables $f:\{-1,1\}^{N-1}\to \bR$ 
		\[\begin{split}
		\delta_i \langle f\rangle^{[i]} &= \big(\delta_i \langle f\rangle^{[i]}\big)(t) = \frac12 \langle f\rangle^{[i]}(1,t) - \frac12 \langle f\rangle^{[i]}(-1,t), \\
		\epsilon_i \langle f\rangle^{[i]}  &= \big(\epsilon_i \langle f\rangle^{[i]}\big) (t) = \frac12 \langle f\rangle^{[i]}(1,t) + \frac12 \langle f\rangle^{[i]}(-1,t), \\
		 \Delta_i \langle f\rangle^{[i]}  &= \big(\Delta_i \langle f\rangle^{[i]}\big) (t) = \big(\epsilon_i \langle f\rangle^{[i]}\big) (t) -   \langle f\rangle^{(i)}. 
		\end{split}\]
The next lemmas estimate certain correlation functions, conditionally on $\sigma_i$. To ease the notation, we write $ \sum$ and understand implicitly that all variables are averaged over $[N]\setminus \{i\}$. 
		
\begin{lemma} \label{lm:delta-m-estimate}
Assume $(\beta,\mu_h,\sigma_h^2)\in \cA_{RS^-}$ and $h\sim \cN(\mu_h,\sigma_h^2)$. Then, for every $k \geq 0,p \geq 1$ with $2k+p \geq 1$, there exists a constant $C>0$ such that
		\[\begin{split}
		\sup_{t\in [0,1]} N^{-p}\sum_{j_1,\ldots, j_p  }  \big\| \big(\Delta_i m_{j_1\ldots j_p;k}^{[i]}\big)(t) \big\|^2_2  &\leq C N^{-A_{2k+p+2}},\\
		\sup_{t\in [0,1]} N^{-p}\sum_{j_1,\ldots, j_p  } \big\| \big(\delta_i m_{j_1\ldots j_p;k}^{[i]} \big)(t)\big\|_2^2 &\leq C N^{-A_{2k+p+1}}.  
		\end{split}\]
\end{lemma}
\begin{proof}
Applying It\^{o}'s lemma w.r.t. $ \big( (g_{ij}(t) )_{j\in [N]:j\neq i}\big)_{t\in [0,1]}$, we find by definition \eqref{eq:corrkp} that
		\begin{align*}
		&\dif \big( \Delta_i m_{j_1\ldots j_p;k}^{[i]}\big) = \beta\sum_{l}  \big(\delta_i m_{lj_1\ldots j_p;k}^{[i]}\big) \dif g_{il}+ \frac{\beta^2}2(1-N^{-1}) \big(\epsilon_i m_{j_1\ldots j_p;k+1}^{[i]} \big)\dif t\,,\\
		&\dif \big( \delta_i m_{j_1\ldots j_p;k}^{[i]}\big) = \beta\sum_{l} \big(\epsilon_im_{lj_1\ldots j_p;k}^{[i]}\big) \dif g_{il}+ \frac{\beta^2}2(1-N^{-1}) \big(\delta_im_{j_1\ldots j_p;k+1}^{[i]}\big)\dif t\,.
		\end{align*}	
Notice that the integrands on the r.h.s. in the previous equations are linear combinations of suitable $(k,p)$-point functions with interaction coupling and external field parameters all satisfying \eqref{eq:AT} \& \eqref{eq:AT+} for $N$ large (by the remarks around \eqref{eq:bminb'}) s.t. by Lemma \ref{corrbound}
		\[\begin{split}
		 \sup_{t\in [0,1]} N^{-p-1}\sum_{l,j_1,\ldots, j_p  } \big\| \big( \epsilon_im_{lj_1\ldots j_p;k}^{[i]} \big)(t)\big\|_2^2 & \leq C N^{-A_{2k+p+1}}, \\
		  \sup_{t\in [0,1]} N^{-p}\sum_{j_1,\ldots, j_p  } \big\| \big( \delta_im_{j_1\ldots j_p;k+1}^{[i]}\big)(t)\big\|_2^2 &\leq C N^{-A_{2k+2+p}}. \end{split}\] 
Hence, employing the It\^{o} isometry yields
		\[\sup_{t\in [0,1]}   N^{-p}\sum_{j_1,\ldots, j_p  }  \big\| \big( \delta_i m_{j_1\ldots j_p;k}^{[i]}\big)(t)\big\|_2^2 \leq C N^{-A_{2k+p+1}}\]
and plugging this into the dynamical equation for $\Delta_i m_{j_1\ldots j_p;k}^{[i]}$, we get
		\[ \sup_{t\in [0,1]} N^{-p}\sum_{j_1,\ldots, j_p  }  \big\| \big(\Delta_i m_{j_1\ldots j_p;k}^{[i]}\big)(t) \big\|^2_2  \leq C N^{-A_{2k+p+2}}. \]
\end{proof}

\begin{lemma}\label{g-delta-est}
Assume $(\beta,\mu_h,\sigma_h^2)\in \cA_{RS^-}$ and $h\sim \cN(\mu_h,\sigma_h^2)$. Then, for every $k \geq 0,p \geq 1$ with $2k+p \geq 1$, there exists a constant $C>0$ such that
		\be\label{eq:g-delta-est}\begin{split}
		\sup_{t\in [0,1]} N^{-p} \sum_{j_1,\ldots,j_p} \Big\| \sum_{j} g_{ij}(t) \big(\Delta_i m_{j j_1\ldots j_p ;k}^{[i]}\big)(t) \Big\|_2^2  \leq C N^{-A_{2k+p+3}}, \\
		\sup_{t\in [0,1]} N^{-p} \sum_{j_1,\ldots,j_p} \Big\| \sum_{j} g_{ij}(t) \big(\delta_i m_{j j_1\ldots j_p ;k}^{[i]}\big) (t) \Big\|_2^2  \leq C N^{-A_{2k+p+2}}.
		\end{split} \ee
\end{lemma}
\begin{proof}
The bounds follow from integration by parts applied to $ g_{ij}(t) $ and $ g_{ij'}(t) $ in
		\[N^{-p} \sum_{j_1,\ldots,j_p}   \sum_{j,j'}  \bE\, g_{ij} g_{ij'} \big(\Delta_i m_{j j_1\ldots j_p ;k}^{[i]}\big) \big(\Delta_i m_{j' j_1\ldots j_p ;k}^{[i]}\big)\]
and applying Lemma \ref{lm:delta-m-estimate}. It will be clear that the bounds are uniform in $t\in [0,1]$ so for simplicity, let us ignore the $t$-dependence in the notation. Applying Gaussian integration by parts first to the $g_{ij'}$ implies
		\[\begin{split}
		&  N^{-p} \sum_{j_1,\ldots,j_p}   \sum_{j,j'}  \bE\, g_{ij} g_{ij'} \big(\Delta_i m_{j j_1\ldots j_p ;k}^{[i]}\big) \big(\Delta_i m_{j' j_1\ldots j_p ;k}^{[i]}\big) \\
		& = N^{-p} \sum_{j_1,\ldots,j_p=1}^N\bigg(  \frac tN \sum_{j=1}^N \bE   \big(\Delta_i m_{j j_1\ldots j_p ;k}^{[i]}\big)^2 \\
		&\hspace{3cm}+ \beta^2 t(1-N^{-1})  \sum_{j=1}^N  \bE\, g_{ij}  \big(\Delta_i m_{j j_1\ldots j_p ;k}^{[i]}\big) \big(\Delta_i m_{ j_1\ldots j_p ;k+1}^{[i]}\big)\\
		&\hspace{3cm} + \beta^2  \frac tN \sum_{j,j'=1}^N  \bE\, g_{ij}  \big(\Delta_i m_{j j' j_1\ldots j_p ;k}^{[i]}\big) \big(\Delta_i m_{j' j_1\ldots j_p ;k}^{[i]}\big) \bigg)\\
		\end{split}\]
and then, applying Gaussian integration by parts to the $g_{ij}$, we arrive at	
		\[\begin{split}
		&  N^{-p} \sum_{j_1,\ldots,j_p}   \sum_{j,j'}  \bE\, g_{ij} g_{ij'} \big(\Delta_i m_{j j_1\ldots j_p ;k}^{[i]}\big) \big(\Delta_i m_{j' j_1\ldots j_p ;k}^{[i]}\big) \\
		& = N^{-p} \sum_{j_1,\ldots,j_p}\bigg(  \frac tN \sum_{j} \bE   \big(\Delta_i m_{j j_1\ldots j_p ;k}^{[i]}\big)^2 + \beta^4 t^2 (1-N^{-1})   \bE\,    \big(\Delta_i m_{j_1\ldots j_p ;k+1}^{[i]}\big)^2 \\
		&\hspace{3cm} + \frac{\beta^4 t^2}N (1-N^{-1})\sum_{j}  \bE\,   \big(\Delta_i m_{j j_1\ldots j_p ;k}^{[i]}\big) \big(\Delta_i m_{j j_1\ldots j_p ;k+1}^{[i]}\big)\\
		&\hspace{3cm}+  \frac{\beta^4t^2}{N}  (1-N^{-1})\sum_{j'}  \bE\,   \big(\Delta_i m_{j' j_1\ldots j_p ;k+1}^{[i]}\big) \big(\Delta_i m_{j' j_1\ldots j_p ;k}^{[i]}\big) \\
		&\hspace{3cm}+   \frac{\beta^4t^2}{N^2}  \sum_{j,j'}  \bE\,   \big(\Delta_i m_{j j' j_1\ldots j_p ;k}^{[i]}\big)^2   \bigg).
		\end{split}\]
Applying Lemma \ref{lm:delta-m-estimate} and Cauchy-Schwarz, we thus obtain
		\[\begin{split}
		\sup_{t\in [0,1]} N^{-p} \sum_{j_1,\ldots,j_p}\Big\| \sum_{j} g_{ij}(t) \big(\Delta_i m_{j j_1\ldots j_p ;k}^{[i]}\big)(t) \Big\|_2^2  \leq CN^{-A_{2k+p+3}}. 
		\end{split} \]
Repeating the computation while replacing $\Delta_i$ by $\delta_i$ in all integrands, we obtain \eqref{eq:g-delta-est}.
\end{proof}

%%%%%%%%%%%%%%%%%%%%%%%%%%%%%%%%%%%%%%%%%%%%%%%%%
%%%%%%%%%%%%%%%%%%%%%%%%%%%%%%%%%%%%%%%%%%%%%%%%%
%%%%%%%%%%%%%%%%%%%%%%%%%%%%%%%%%%%%%%%%%%%%%%%%%

\section{Proof of Propositions \ref{prop:hTAP} and \ref{prop:M-H}} \label{sec:proofpart1}

We start with the proof of Prop.\ \ref{prop:hTAP} which is a simple consequence of Lemma \ref{lm:delta-m-estimate}. 
\begin{proof}[Proof of Prop.\ \ref{prop:hTAP}]
A direct computation shows that 
		\[ m_{ij} = (1 - m_i^2)\,\delta_i m_j^{[i]}.  \]
Now, applying It\^{o} as in Lemma \ref{lm:delta-m-estimate} in the previous section on the $i$-th row of $\textbf{G}$, we find
		\[\begin{split}
		\frac{m_{ij}}{1 - m_i^2} -  \beta \sum_{l } g_{il} m_{lj}^{(i)}  &=\beta  \sum_{k}\int_0^1 \big(\Delta_{i}m_{jk}^{[i]}\big)(t)\dif g_{ik}(t) + \frac {\beta^2}{2 N}\sum_{k}\int_0^1 \big(\delta_{i} m_{kkj}^{[i]}\big)(t)\dif t \\
		& = \beta\sum_{k}\int_0^1 \big(\Delta_{i}m_{jk}^{[i]}\big)(t)\dif g_{ik}(t) + \frac{\beta^2}2(1-N^{-1}) \int_0^1 \big(\delta_{i} m_{j;1}^{[i]}\big)(t)\dif t.
		\end{split} \]
Using that $ (1 - m_i^2)\leq 1$ and applying Lemma \ref{lm:delta-m-estimate}, we conclude that
		\be \label{eq:2435}
		\begin{split}
		&\Big\| \frac{m_{ij}}{1 - m_i^2} - \beta\sum_{l \not = i} g_{il} m_{lj}^{(i)} \Big\|_2^2 \\
		&\leq  \frac CN \sum_{k}\sup_{t\in[0,1]}  \big\| \big(\Delta_{i}m_{jk}^{[i]}\big)(t) \big\|_2^2  + C \sup_{t\in[0,1]}  \big\| \big(\delta_{i} m_{j;1}^{[i]}\big)(s) \big\|_2^2  \\
		& = C \sup_{t\in[0,1]} \bigg( \frac1{(N-1)^2} \sum_{j_1,j_2}  \big\| \big(\Delta_{i}m_{j_1j_2}^{[i]}\big)(t) \big\|_2^2  + \frac1{N-1} \sum_{j}  \big\| \big(\delta_{i} m_{j;1}^{[i]}\big)(s) \big\|_2^2\bigg)  \leq C  N^{-5/2}.
		\end{split}\ee
Notice that in the second step of the last bound we used the symmetry among the sites $j \in [N]\setminus\{i\}$, in order to average over both indices of the two point functions.
\end{proof}
	
We conclude this section with the proof of Prop.\ \ref{prop:M-H}. Our goal is to show that 
		\[ \tbf Y = \big(\tbf D  + \beta^2(1-q)   -  \beta \tbf G \big) \tbf M  - \text{id}_{\bR^{N }} \in \bR^{N\times N}\]
with entries   
		\be \label{eq:defyik}  
		y_{ik} =  \begin{cases} (1-m_i^2) \big( \beta^2(1-q) - \beta \sum_j g_{ij} \delta_im_j^{[i]}\big) &: i=k, \\  (1-m_i^2)^{-1}m_{ik} + \beta^2(1-q)m_{ik}  - \beta\sum_{j} g_{ij}m_{jk} &: i\neq k,  \end{cases}
		\ee
has operator norm $ \| \tbf Y\|_{\text{op}}$ bounded by $ o(1) \|\tbf M\|_{\text{op}}$, up to a quantity $O(1)$ which is of order one with high probability. Our proof relies on Prop.\ \ref{prop:hTAP}. In the following, the standard Frobenius (or Hilbert-Schmidt) norm is denoted by $\|\cdot\|_{\text F}$; recall that $ \| \cdot\|_{\text{op}}\leq \|\cdot\|_{\text F}$. 

\begin{proof}[Proof of Prop.\ \ref{prop:M-H}]
A direct computation shows that 
		\[\begin{split}
		m_{jk} =  \langle m_j^{[i]};m_k^{[i]}\rangle + \langle m_{jk}^{[i]}\rangle & =\big( \delta_im_j^{[i]}\big) m_{ik} + \epsilon_i m_{jk}^{[i]} + m_i \big(\delta_i m_{jk}^{[i]}\big) 
		\end{split} \]
so that, motivated by Prop.\ \ref{prop:hTAP}, we can decompose $ y_{ik}$ for $i\neq k$ into 
		\[\begin{split}
		y_{ik} & = \frac{m_{ik}}{1-m_i^2} - \beta\sum_j g_{ij} m_{jk}^{(i)}  - \beta\sum_{j} g_{ij} \big(m_{jk} -m_{jk}^{(i)}\big) + \beta^2(1-q)m_{ik} \\
		& = \Big( \beta^2 (1-q)  - \beta \sum_j g_{ij} \delta_im_j^{[i]}  \Big) m_{ik} + \Big( \frac{m_{ik}}{1-m_i^2} - \beta\sum_j g_{ij} m_{jk}^{(i)}\Big)  \\
		&\hspace{0.5cm}  - \beta \sum_j g_{ij} \big(\Delta_i m_{jk}^{[i]} \big)  - \beta m_i \sum_jg_{ij}\big( \delta_i m_{jk}^{[i]}\big).
		\end{split}\]
Comparing this with the diagonal entries of $\tbf Y$ in Eq.\ \eqref{eq:defyik}, we can split
		\[\tbf Y = \tbf Y_1 \tbf M + \tbf Y_2 - \tbf Y_3 - \tbf Y_4, \]
where the $ \tbf Y_j\in \bR^{N\times N}$, for $j\in \{1,2,3,4\}$, are defined by
		\[\begin{split}
		( \tbf Y_1 )_{ik}  &= \Big( \beta^2 (1-q)  - \beta \sum_j g_{ij} \delta_im_j^{[i]}  \Big) \delta_{ik} ,
		\hspace{0.5cm}(\tbf Y_2)_{ik}  = \begin{cases} 0 &: i=k, \\  \frac{m_{ik}}{1-m_i^2} - \beta\sum_j g_{ij} m_{jk}^{(i)} &: i\neq k,  \end{cases} \\
		(\tbf Y_3)_{ik} & = \begin{cases} 0 &: i=k, \\  \beta \sum_j g_{ij} \big(\Delta_i m_{jk}^{[i]} \big)  &: i\neq k,  \end{cases} 
		\hspace{1.1cm}(\tbf Y_4)_{ik}  = \begin{cases} 0 &: i=k, \\  \beta m_i \sum_jg_{ij}\big( \delta_i m_{jk}^{[i]}\big)  &: i\neq k.  \end{cases}
		\end{split}\]
The matrix $ \tbf Y_2$ vanishes in norm when $N\to \infty$, in the sense of probability, by Eq.\ \eqref{eq:2435} and Markov's inequality: for every $\delta >0$, we find that
		\[ \bP\big( \| \tbf Y_2\|_{\text{op}} > \delta\big)\leq \delta^{-2} \bE \| \tbf Y_2\|_{\text{op}}^2 \leq \delta^{-2} \bE \| \tbf Y_2\|_{\text{F}}^2 = N^2  \max_{i,j\in[N]: i\neq j} \bE  \big(\tbf Y_2)_{ij}^2 \leq C \delta^{-2}N^{-1/2}.\]	
By Lemma \ref{g-delta-est}, the same argument implies $ \lim_{N\to\infty} \| \tbf Y_3\|_{\text{op}} =0 $, and we have that
		\[\begin{split}
		\bP \big ( \| \tbf Y_4\|_{\text{op}} > K\big) &\leq K^{-2} \beta^2 N^2 \max_{i,j\in[N]: i\neq j} \bE  \big(\tbf Y_4)_{ij}^2 \\
		&\leq C K^{-2} N^2    N^{-1}\sum_{j_1} \Big\| \sum_jg_{1j}\big( \delta_1 m_{jj_1}^{[1]}\big)\Big\|_2^2  \leq C K^{-2}, 
		\end{split}\]	
for some $C>0$, independent of $N$, where we used symmetry among the sites to reduce the computation to the first row of $\tbf G$. Combining these remarks, Prop.\ \ref{prop:M-H} follows if
		\be\label{eq:laststep1.2}\lim_{N\to\infty} \|\tbf Y_1 \|_{\text{op}} = 0 \ee		
in probability. 

To prove \eqref{eq:laststep1.2}, we argue first as in \eqref{eq:2435} to obtain that
		\[ \bE \Big|  \sum_j g_{ij} \Big( \delta_im_j^{[i]} -   \beta \sum_{k} g_{ik}m_{kj}^{(i)}\Big) \Big|^2 \leq C N^{-5/2}. \]
Indeed, this bound is a direct consequence of It\^o's lemma, which shows that
		\[\begin{split}
		 \dif \sum_j g_{ij} \Big( \delta_im_j^{[i]} -   \beta \sum_{k} g_{ik}m_{kj}^{(i)}\Big) 
		& =   \beta\!\sum_{j}\Big( \delta_im_j^{[i]} \!- \!  \beta \sum_{k} g_{ik}m_{kj}^{(i)}\Big) \dif g_{ij}\!+ \!\beta\sum_{j, k} g_{ij} \big( \Delta_{i}m_{jk}^{[i]}\big) \!\dif g_{ik} \\
		&\hspace{0.5cm}+ \frac{\beta^2}2(1-N^{-1})\sum_j g_{ij}  \big(\delta_{i} m_{j;1}^{[i]}\big)\dif t +  \frac\beta N \sum_{j }   \big( \Delta_{i}m_{jk}^{[i]}\big) \dif t,
		\end{split}\]
and estimating the different contributions on the r.h.s. of the previous identity using Lemmas \ref{lm:delta-m-estimate} and \ref{g-delta-est}. Combining the previous bound with Markov thus implies 
		\[\lim_{N\to\infty} \sup_{i\in[N]} \Big| (\tbf Y_1 )_{ii}  - \Big( \beta^2(1-q) -\beta^2 \sum_{j}g_{ij} \sum_{k} g_{ik}m_{kj}^{(i)}\Big)\Big| =0    \]
in probability. 

We now  split  
		\[\begin{split}
		& \beta^2(1-q) -\beta^2 \sum_{j}g_{ij} \sum_{k} g_{ik}m_{kj}^{(i)} \\
		&= \beta^2 \big\langle R_{1,2} - q \big\rangle^{(i)} - \beta^2\sum_{j} \big( g_{ij}^2 - (N-1)^{-1}\big) \big(1- \big(m_j^{(i)}\big)^2\big)-\beta^2 \sum_{j,k:j\neq k}g_{ij} g_{ik}m_{kj}^{(i)}, 
		\end{split}\]
and it is straightforward to deduce from Theorem \ref{thm:conc} that 
		\[ \lim_{N\to\infty} \sup_{i\in[N]} \big| \big\langle R_{1,2} - q \big\rangle^{(i)}\big|  = 0 \]
in probability. Indeed, this follows from Markov's inequality combined with Theorem \ref{thm:conc}, the continuity of $ (\beta, \mu_h, \sigma_h^2)\to q_{\beta,h} =q_{\beta,\mu_h, \sigma_h^2} $ and the fact that $ \langle \cdot\rangle^{(i)}$ is a SK measure with interaction coupling $ \beta' =\beta'_N$ and Gaussian external field $h'=h'_N$ such that $  | \beta- \beta'| \to 0$, $ | \sigma_h^2-\sigma_{h'}^2| \to 0$ as $N\to \infty$ and $|\mu_h-\mu_{h'}| =0$, as remarked already at \eqref{eq:bminb'}.

On the other hand, a standard exponential concentration bound yields 
		\[ \lim_{N\to\infty} \sup_{i\in[N]} \Big| \sum_{j} \big( g_{ij}^2 - (N-1)^{-1}\big) \big(1- \big(m_j^{(i)}\big)^2\big)\Big|= \lim_{N\to\infty} \sup_{i\in[N]} \Big| \sum_{j} \big( g_{ij}^2 - N^{-1}\big) \Big|=0\]
and, finally, applying once more Theorem \ref{thm:conc}, a straightforward computation shows that
		\[\begin{split}
		\bE \Big(  \sum_{j,k:j\neq k}g_{ij} g_{ik}m_{kj}^{(i)}\Big)^4 &  \leq CN^{-4} \sum_{j_1,j_2,j_3,j_4}  \bE \, \big(m_{j_1j_2}^{(i)}\big)^2 \big(m_{j_3j_4}^{(i)}\big)^2.
		\end{split}\]
Then, using the identity 
		\[\begin{split}
		\big(m_{j_1j_2}^{(i)}\big)^2 &= \big\langle  \big( \sigma_{j_1}^1 -  \sigma_{j_1}^2 \big)\big( \sigma_{j_2}^1 -  \sigma_{j_2}^3 \big)\big( \sigma_{j_1}^{-1} -  \sigma_{j_1}^{-2} \big)\big( \sigma_{j_2}^{-1} -  \sigma_{j_2}^{-3} \big)  \big \rangle^{(i)}
		\end{split}\]
together with Theorem \ref{thm:conc} and Cauchy-Schwarz, we arrive at
		\[\begin{split}
		\bE \Big(  \sum_{j,k:j\neq k}g_{ij} g_{ik}m_{kj}^{(i)}\Big)^4 &  \leq C \,\bE\, \big( \big\langle  \big( R_{1,-1} +R_{2,-2} - R_{1,-2} - R_{2,-1}\big)^2 \big\rangle^{(i)} \big)^2 \leq C N^{-2}.
		\end{split}\]			
Hence, by Markov's inequality, we obtain that
		\[ \lim_{N\to\infty} \sup_{i\in[N]} \Big|    \sum_{j,k:j\neq k}g_{ij} g_{ik}m_{kj}^{(i)}\Big| = 0, \]
in probability. This proves $\lim_{N\to\infty} \|\tbf Y_1 \|_{\text{op}} = 0$ and concludes Prop.\ \ref{prop:M-H}.
\end{proof}

%%%%%%%%%%%%%%%%%%%%%%%%%%%%%%%%%%%%%%%%%%%%
%%%%%%%%%%%%%%%%%%%%%%%%%%%%%%%%%%%%%%%%%%%%
%%%%%%%%%%%%%%%%%%%%%%%%%%%%%%%%%%%%%%%%%%%%

\section{Proof of Proposition \ref{prop:op-H}} \label{sec:AMP}

In this section we conclude Theorem \ref{thm:main} by proving Prop.\ \ref{prop:op-H}. Our arguments rely on the main results of \cite{Bolt1, Bolt2,CT21} and translate, up to a few modifications, the main ideas from \cite{celentano2022sudakov} to the present context. For the rest of this paper, $h>0$ denotes a deterministic field strength as in the definition of $H_N$ in Eq.\ \eqref{eq:HN} in Section \ref{sec:intro}.

We first recall the main result of \cite{CT21}, which states that, in a suitable subregion of the replica symmetric phase, the magnetization vector $\tbf m $ is well approximated by an iterative solution to the TAP \cite{TAP} equations introduced by Bolthausen in \cite{Bolt1, Bolt2}. To state this precisely, we need to recall Bolthausen's construction of the TAP solution and collect some of its properties. Here, we find it convenient to follow the conventions and notation used in \cite{Bolt2, BY}. We start with the sequences $(\alpha_{k})_{k\in\mathbb{N}}$, $(\gamma_k)_{k\in\mathbb{N}}$ and $(\Gamma_k)_{k\in\mathbb{N}}$ which have initializations
		\[ \alpha_1 = \sqrt{q}\gamma_1,\hspace{0.5cm} \gamma_1 = E \tanh(h+ \beta\sqrt q Z ), \hspace{0.5cm} \Gamma_1^2 = \gamma_1^2\]
and we define $\psi: [0,q]\to [0,q]$ by
		\[ \psi(t)=E \tanh\big(h+ \beta \sqrt t Z +\beta \sqrt{q-t}Z' \big) \tanh \big(h+ \beta \sqrt t Z +\beta \sqrt{q-t}Z'' \big) \]
Then, we set recursively
		\[\alpha_{k} =  \psi(\alpha_{k-1}), \hspace{0.5cm} \gamma_k = \frac{\alpha_{k} - \Gamma^2_{k-1}}{\sqrt{q - \Gamma^2_{k-1}}}, \hspace{0.5cm} \Gamma_k^2 = \sum_{j=1}^k \gamma_j^2. \]
\begin{lemma}{(\cite[Lemma 2.2, Corollary 2.3, Lemma 2.4]{Bolt1}, \cite[Lemma 2]{Bolt2})}\label{lm:seqlemma}
\begin{enumerate}[1)]
\item $\psi$ is strictly increasing and convex in $[0,q]$ with $0< \psi(0) <\psi(q) = q$. If \eqref{eq:AT} is satisfied, then $q$ is the unique fixed point of $\psi$ in $[0,q]$.
\item $(\alpha_{k})_{k\in\mathbb{N}}$ is increasing and $\alpha_{k} >0$ for all $k\in\mathbb{N}$. If \eqref{eq:AT} is satisfied, then $\lim_{k\to\infty} \alpha_{k} =q$. 
\item For $k \geq 2$, we have that $ 0< \Gamma_{k-1}^2<\alpha_{k} <q $ and that $0<\gamma_k < \sqrt{q -\Gamma_{k-1}^2}$. If \eqref{eq:AT} is satisfied, then $\lim_{k\to\infty} \Gamma_k^2 =q$ and, consequently, $\lim_{k\to\infty} \gamma_k =0$. 
\end{enumerate}
\end{lemma}

Next, we recall Bolthausen's decomposition of the interaction matrix $\tbf G$, which yields a convenient representation of the iterative TAP solution. Without loss of generality, we may assume that the interaction matrix $\tbf G \in \bR^{N\times N}$ in \eqref{eq:HN} is equal to
		\[ \tbf G  = \frac{ \tbf W + \tbf W^T }{\sqrt 2} = \bar {\tbf W}  \]
for a random matrix $\tbf W = (w_{ij})_{1\leq i,j\leq N} \in \bR^{N\times N} $ with zero diagonal and i.i.d. Gaussian entries $w_{ij}\sim \cN(0,N^{-1})$ (without symmetry constraint). Here and in the following, we abbreviate $ \bar {\tbf X} = (\tbf X + \tbf X^T)/\sqrt{2}$ for $\tbf X \in \bR^{N\times N}$. Then, we set
		\[ \tbf{W}^{(1)} = \tbf{W}, \hspace{0.5cm} \tbf{G}^{(1)} =  \bar{ \tbf W}^{(1)} , \hspace{0.5cm} \phi^{(1)} = N^{-1/2}\, \tbf{1}  \in \mathbb{R}^N, \hspace{0.5cm} \tbf{m}^{(1)}=\sqrt{q}\tbf{1}\in\mathbb{R}^N \]
and, assuming $\tbf{W}^{(s)}, \tbf{G}^{(s)} = \bar{ \tbf W}^{(s)}, \phi^{(s)}, \tbf{m}^{(s)}$ are defined for $1\leq s\leq k$, we set\footnote{Notice that here we use the convention that $ (\phi^{(s)})_{s=1}^k$ forms an orthonormal sequence in $(\bR^N, (\cdot,\cdot) )$. In contrast to that, in \cite{Bolt2, BY} the inner product $\langle\cdot,\cdot\rangle $ and the tensor product $\otimes$ are rescaled by a factor $N^{-1}$.}
		\[ \zeta^{(s)} =  \tbf{G}^{(s)}  N^{1/2} \phi^{(s)}.\]
Moreover, $\cG_k$ denotes the $\sigma$-algebra
		\[ \cG_k = \sigma\big( \tbf{W}^{(s)} N^{1/2} \phi^{(s)},(\tbf{W}^{(s)})^T  N^{1/2} \phi^{(s)}: 1\leq s\leq k   \big). \]
Conditional expectations and respectively probabilities with respect to $\cG_k$ are denoted by $\bE_k$ and $\bP_k$. We then define the iterative cavity field $\tbf{z}^{(k+1)} \in \bR^N$ and the iterative TAP solution $\tbf{m}^{(k+1)}\in \bR^N$ at step $k+1$ by
		\begin{equation}\begin{split}\label{eq:AMP}
		\tbf{z}^{(k+1)} &=   \sum_{s=1}^{k-1}\gamma_s\zeta^{(s)} + \sqrt{q-\Gamma_{k-1}^2}\zeta^{(k)},\hspace{0.5cm} \tbf{m}^{(k+1)} =\tanh \big(h\tbf{1}+\beta\, \tbf{z}^{(k+1)}\big)
		 \end{split}\end{equation}
and we also set
		\[\phi^{(k+1)}  = \frac{\tbf{m}^{(k+1)} - \sum_{s=1}^{k}   ( \tbf{m}^{(k+1)}, \phi^{(s)}) \phi^{(s)}}{\big\| \tbf{m}^{(k+1)} - \sum_{s=1}^{k}   ( \tbf{m}^{(k+1)}, \phi^{(s)} )\phi^{(s)} \big\|},\]
recalling that $\phi^{(k+1)}$ is well-defined for all $k<N$ \cite[Lemma 5]{Bolt2}. Finally, we define
		\[\begin{split}
		\tbf{W}^{(k+1)} =&\, \tbf{W}^{(k)}- \rho^{(k)}, \hspace{0.5cm}\text{for}  \\
		 \rho^{(k)} =&\,    \tbf{W}^{(k)}\phi^{(k)}\otimes \phi^{(k)} +  \phi^{(k)}\otimes (\tbf{W}^{(k)})^T\phi^{(k)} - ( \tbf{W}^{(k)}\phi^{(s)} , \phi^{(k)})\, \phi^{(k)}\otimes \phi^{(k)},
		\end{split}\]
where $ (\tbf x \otimes \tbf y) (\tbf z) = (\tbf y, \tbf z) \,\tbf x$, and by symmetrization
		\[\begin{split}
		 \tbf{G}^{(k+1)} &= \bar{\tbf{W}}^{(k+1)} =  \tbf{G}^{(k)}- \bar{ \rho}^{(k)}, \hspace{0.5cm}\text{for}\\
		\bar{ \rho}^{(k)}&=   N^{-1/2}\zeta^{(k)}\otimes \phi^{(k)} +  N^{-1/2}\phi^{(k)}\otimes \zeta^{(k)} - N^{-1/2} (\zeta^{(k)}, \phi^{(k)} ) \,\phi^{(k)}\otimes \phi^{(k)}.
		\end{split}\]
$(\phi^{(s)})_{s=1}^k$ is orthonormal in $(\mathbb{R}^N, (\cdot,\cdot))$ and $\tbf{P}^{(k)}$, $\tbf{Q}^{(k)}$ denote the orthogonal projections
		\[ \tbf{P}^{(k)} = \sum_{s=1}^{k} \phi^{(s)}\otimes \phi^{(s)} = (P^{(k)}_{ij})_{1\leq i,j\leq N}, \hspace{0.5cm} \tbf{Q}^{(k)}= \tbf{1}_{\bR^N}-\tbf{P}^{(k)}= (Q^{(k)}_{ij})_{1\leq i,j\leq N}.  \]
Notice that $ \tbf{P}^{(k)} $ equals the orthogonal projection onto 
		$$ \text{span} \,\big( \tbf m^{(s)}: 1\leq s\leq k \big) = \text{span} \,\big( \phi^{(s)}: 1\leq s\leq k \big).$$
In the next result we collect several useful facts about $ (\tbf z^{(s)})_{s=1}^k, (\tbf m^{(s)})_{s=1}^k $ and $ (\tbf G^{(s)})_{s=1}^k $. We say for $ (X_N)_{N\geq 1}, (Y_N)_{N\geq 1}$ that may depend on parameters like $\beta, h,$ etc. that 
		$$ X_N\simeq Y_N$$
if and only if there exist positive constants $c, C >0$, which may depend on the parameters, but which are independent of $N $, such that for every $t >0$ we have 
	$$ \mathbb{P} ( |X_N-Y_N| > t ) \leq C e^{-cNt^2}.$$
\begin{proposition}{(\cite[Prop. 2.5]{Bolt1}, \cite[Prop. 4, Prop. 6, Lemmas 3, 11, 14 \& 16]{Bolt2})} \label{prop:Bolt}
\begin{enumerate}[1)]
\item $\tbf{m}^{(k)}$ and $\phi^{(k)}$ are $\cG_{k-1}$-measurable for every $k\geq 1$ and 
		\[ \emph{\tbf{W}}^{(k)}\phi^{(s)} =  (\emph{\tbf{W}}^{(k)})^{T}\phi^{(s)}= \emph{\tbf{G}}^{(k)}\phi^{(s)}  =0, \,\,\forall \,\,s<k.  \]
\item Conditionally on $\cG_{k-2}$, $ \emph{\tbf{W}}^{(k)}$ and $ \emph{\tbf{W}}^{(k-1)}$ are Gaussian, and we have that
		\[\bE_{k-2} \, w_{ij}^{(k)}w_{st}^{(k)} = \frac1N  Q^{(k-1)}_{is} Q^{(k-1)}_{jt} \]
and, consequently, that
		\[\bE_{k-2} \, g_{ij}^{(k)}g_{st}^{(k)} = \frac1N  Q^{(k-1)}_{is} Q^{(k-1)}_{jt} + \frac1N  Q^{(k-1)}_{it} Q^{(k-1)}_{js}. \]
\item Conditionally on $\cG_{k-2}$, $\emph{\tbf{W}}^{(k)}$ is independent of $\cG_{k-1}$. In particular, conditionally on $\cG_{k-1}$, $\emph{\tbf{W}}^{(k)}$ and $\emph{\tbf G}^{(k)}$ are Gaussian with the same covariance as in 2).
\item For every $k\geq 1$, $ N^{-1/2}(\zeta^{(k)}, \phi^{(k)}) $ is unconditionally Gaussian with variance $2/N$.
\item Conditionally on $\cG_{k-1}$, the random variables $\zeta^{(k)}$ are Gaussian with 
		\[ \bE_{k-1} \zeta^{(k)}_i\zeta^{(k)}_j = Q^{(k-1)}_{ij} +   \phi^{(k)}_i\phi^{(k)}_j.\]
\item For every $k\geq 1$ and $  s < k$, one has
		\[\begin{split}
		 N^{-1/2} ( \tbf{m}^{(k)}, \phi^{(s)})&\simeq \gamma_s,\hspace{0.2cm}N^{-1/2} ( \tbf{m}^{(k)}, \phi^{(k)}) \simeq \sqrt{q-\Gamma_{k-1}^2},\\
		N^{-1} (  \tbf{m}^{(k)}, \tbf{m}^{(s)}) &\simeq \alpha_s, \hspace{0.2cm} N^{-1} (  \tbf{m}^{(k)},\tbf{m}^{(k)} ) \simeq q.
		 \end{split} \]
In particular, $ \lim_{N\to\infty} N^{-1} \bE\, \| \tbf m^{(k)} \|^2 =   q $ and, assuming \eqref{eq:AT}, we have that
		\[ \begin{split}
		\limsup_{j,k\to \infty }  \lim_{N\to\infty}  N^{-1} \bE\, \| \tbf m^{(j)}- \tbf m^{(k)} \|^2=0.
		\end{split} \]
\item For every $k\geq 1$, we have that
		\[ \begin{split}
		\big\| N^{-1/2} \big(\zeta^{(k)}, \tbf m^{(k+1)} \big) - \beta (1-q) \sqrt{q- \Gamma_{k-1}^2}  \big\|_2 = 0,  
		\end{split}\]
and, for $1\leq s\leq k-1$, we have that
		\[ \begin{split}
		\big\| N^{-1/2} \big(\zeta^{(s)}, \tbf m^{(k+1)} \big) - \beta (1-q) \gamma_s    \big\|_2 = 0. 
		\end{split}\]
\end{enumerate}
\end{proposition}  

In addition to the properties listed in Prop.\ \ref{prop:Bolt} it is well-known that the sequences $(\tbf m^{(s)})_{s=1}^k$ and $(\tbf z^{(s)})_{s=2}^{k+1}$ have an explicit joint limiting law, which we recall in the next proposition. In its statement, convergence of a sequence $ (\mu_n)_{n\in\bN}$ of probability measures on $\bR^d$, $\mu_n\in \cP(\bR^d)$ for each $n\in\bN$, to a limiting measure $\mu \in \cP(\bR^d)$ in $\cW_2(\bR^d) $ means that 
		\[ \lim_{n\to\infty} \cW_2(\mu_n,\mu) =0, \hspace{0.5cm} \text{ where } \hspace{0.5cm} \cW_2(\mu,\nu ) = \inf_{\Pi} \sqrt{ E \|\tbf X-\tbf Y\|^2}\]
denotes the usual Wasserstein $2$-distance between two probability measures $ \mu, \nu \in\cP(\bR^d)$ (the infimum is taken over all couplings $\Pi \in\cP(\bR^{2d})$ of $\mu$ and $\nu$, and $(\tbf X,\tbf Y)$ has joint distribution $(\tbf X,\tbf Y)\sim \Pi$). The next theorem follows from combining \cite{JM13, Bolt1, Bolt2}. 

\begin{thm}\label{thm:law} In the sense of probability (w.r.t. the disorder $\tbf G$), we have that 
		\[ \lim_{N\to\infty } \frac 1 N\sum_{i=1}^N \delta_{m_i^{(1)},\ldots,m_i^{(k)},z_i^{(2)},\ldots,z_i^{(k+1)}} =  \cL_{(M_1,\ldots,M_k,Z_2,\ldots,Z_{k+1})}  \]
in $W_2(\bR^{2k}) $, where $ M_1= \sqrt{q}$, $M_s \stackrel{s>1}{=} \tanh( h+ \beta Z_{s})$, $(Z_2,\ldots, Z_{k+1}) \sim \cN(\mathbf{0},\mathbf K_{\leq k})$ is a Gaussian vector in $\bR^k$ and $\cL_{(M_1,\ldots,M_K,Z_2,\ldots,Z_{k+1})} $ is the law of $ (M_1,\ldots,M_K,Z_2,\ldots,Z_{k+1})$. The covariance $\mathbf K_{\leq k}\in \bR^{k\times k}$ equals $K_{st} = \emph{Cov}(Z_{s+1},Z_{t+1}) = E M_s M_t$; in particular 
		\[\begin{split}
		K_{st}   & \stackrel{s,t >1}{=} E \tanh (h+ \beta Z_{s} ) \tanh (h+ \beta Z_{t}) = E M_s M_t =\begin{cases} q &:  s=t, \\ \alpha_{s\wedge t} &: s\neq t. \end{cases}
		\end{split}\]
\end{thm}
\noindent \textbf{Remark 4.3.} \emph{In \cite{JM13, Bolt1}, the iterative TAP sequence is defined differently compared to \eqref{eq:AMP}. Namely, in the language of \cite{JM13, Bolt1} one sets $ \wt {\tbf{m}}^{(0)}=\tbf 0,  \wt{\tbf{m}}^{(1)} = \sqrt q \tbf 1 $ and for $k\geq 1$
		\[ \wt{\tbf{z}}^{(k+1)} =   \tbf G  \wt{\tbf m}^{(k)}- \beta(1-q)  \wt{\tbf m}^{(k-1)}  ,\hspace{0.5cm} \wt{\tbf{m}}^{(k+1)} =\tanh \big(h\tbf{1}+\beta\, \wt{\tbf{z}}^{(k+1)}\big). \]
The validity of Theorem \ref{thm:law} then follows from 
		\[ \lim_{N\to\infty} N^{-1} \bE\, \big\| \wt {\tbf z}^{(k)}- \tbf z^{(k)}\big \|^2 =0, \hspace{0.5cm}\lim_{N\to\infty} N^{-1} \bE\, \big\| \wt {\tbf m}^{(k)}-\tbf m^{(k)}\big \|^2=0  \]
for every $k\geq 2$, which can be proved inductively (note that $ \wt {\tbf m}^{(1)} = \tbf m^{(1)}, \wt {\tbf m}^{(2)} = \tbf m^{(2)} $), based on Prop.\ \ref{prop:Bolt} and the decomposition $ \tbf G = \tbf G^{(k)} + \sum_{s=1}^{k-1} \bar{\rho}^{(s)} $, and which implies that
		\[\lim_{N\to\infty } \cW_2 \bigg( \frac 1 N\sum_{i=1}^N \delta_{m_i^{(1)},\ldots,m_i^{(k)},z_i^{(2)},\ldots,z_i^{(k+1)}},  \frac 1 N\sum_{i=1}^N \delta_{\wt{m}_i^{(1)},\ldots,\wt{m}_i^{(k)},\wt{z}_i^{(2)},\ldots,\wt{z}_i^{(k+1)}} \bigg)=0 \]
in probability.}		
\vspace{0.2cm}

The properties of Bolthausen's iterative solution $ (\tbf m^{(k)})_{k \geq 1}$ suggest that it is close to the magnetization vector $ \tbf m$ of the SK model, at least in a suitable subregion of the RS phase. This is proved in \cite{CT21}, based on a locally uniform overlap concentration assumption. 

\begin{thm}[{\cite{CT21}}]  \label{thm:CT}
Assume that $\beta,h>0$ are such that for some $\delta>0$ we have
	\begin{equation}\label{eq:CTcond} \lim_{N \to \infty} \sup_{\beta-\delta\leq \beta'\leq \beta} \E \Ex{| R_{12} - q_{\beta',h}|^2}_{\beta',h}  = 0\,. \end{equation}
Then we have that 
	\be \label{eq:CTconc}\lim_{k \to \infty} \lim_{N \to \infty}  N^{-1} \,\bE \,  \big\| \tbf m- \tbf m^{(k)} \big\|^2  = 0.\ee
\end{thm}
\noindent \textbf{Remark 4.4.} \emph{Observe that \eqref{eq:CTcond} is satisfied under \eqref{eq:AT} \& \eqref{eq:AT+}, by Eq.\ \eqref{eq:conc+} in Theorem \ref{thm:conc}. Notice furthermore that \cite{CT21} defines the TAP iteration differently compared to \eqref{eq:AMP}:
		\[  (\mathbf z')^{(k+1)} =\,\tbf G\,(\mathbf m')^{(k)} -\beta \big(1 - N^{-1}\| (\mathbf m')^{(k)}\|^2\big)\,(\mathbf m')^{(k-1)},\hspace{0.3cm}(\mathbf m')^{(k+1)}= \tanh\big(h + \beta (\mathbf z')^{(k+1)} \big).\]
Similar remarks as for Theorem \ref{thm:law} apply: based on Prop.\ \ref{prop:Bolt} and induction, one obtains
		\[ \lim_{N\to\infty} N^{-1} \bE\, \big\| (\tbf z')^{(k)}-  \tbf z^{(k)}\big \|^2 =0, \hspace{0.5cm} \lim_{N\to\infty} N^{-1} \bE\, \big\| (\tbf m')^{(k)}-  \tbf m^{(k)}\big \|^2=0,  \]
which implies \eqref{eq:CTconc} with $\tbf m^{(k)}$ as defined in \eqref{eq:AMP}. }
\vspace{0.2cm} 

Equipped with the above preparations, we now turn to the proof of Prop.\ \ref{prop:op-H}, which is a direct consequence of the next proposition. For a sequence $ (X_N)_{N\in\bN}$, we set
		\[\begin{split} \text{p-}\!\liminf_{N \to \infty} X_N &= \sup \Big\{ t\in\bR\!:  \lim_{N\to\infty} \bP (X_N\leq t) =0\Big\}, \\
		 \text{p-}\!\limsup_{N \to \infty} X_N&= \inf \Big\{ t\in\bR\!:  \lim_{N\to\infty} \bP (X_N \geq t) =0\Big\}. \end{split}\]
The next result follows by translating the main ideas \cite[Section 4]{celentano2022sudakov} to the present context. For completeness, we carry out the key steps, with a few modifications, in detail below.
\begin{proposition} \label{prop:vHv} Assume that $ (\beta,h)$ satisfy \eqref{eq:AT}, and define for $ \tbf v\in (-1,1)^N$ the \linebreak diagonal matrix $ \tbf D_{\tbf v}\in \bR^{N\times N} $ and the set $ S_{N,\eps,k}\subset \bR^{N}\times (-1,1)^N$ by 
		\be\label{eq:defDS} \big( \tbf D_{\tbf v}\big)_{ij} =  \frac{\delta_{ij}}{1-v_i^2}, \hspace{0.2cm} S_{N,\eps,k} = \Big\{(\tbf u, \tbf v)\in \bR^{N}\times (-1,1)^N: \|\tbf u\|^2=1,  \|\mathbf v -\mathbf m^{[k]}\|^2  \leq N\eps\Big\}.\ee
Then, there exists a constant $c = c_{\beta,h}>0$, that is independent of $N$, such that
	\[ \liminf_{\eps \to 0}\;\liminf_{k \to \infty}\;\emph{p-}\liminf_{N \to \infty} \bigg[\inf_{(\tbf u, \tbf v) \in S_{N,\eps,k}} \Big(\tbf u,  \big  ( \tbf D_{\tbf v} +\beta^2(1-q) - \frac{2\beta^2}N \tbf v\otimes \tbf v- \beta \tbf G\big) \tbf u \Big) \bigg] \geq c. \]
\end{proposition}
\noindent \textbf{Remark 4.5.} \emph{As already remarked in Section \ref{sec:intro}, note that Prop.\ \ref{prop:vHv} only requires $ (\beta,h)$ to satisfy the AT condition \eqref{eq:AT}. The additional assumption \eqref{eq:AT+} in Prop.\ \ref{prop:op-H} is used to approximate the magnetization $\tbf m$ of the SK model by the iterative TAP solution $\tbf m^{(k)}$. }
\vspace{0.2cm} 

Before giving the proof of Prop.\ \ref{prop:vHv}, let us first record that it implies Prop.\ \ref{prop:op-H}. 
\begin{proof} [Proof of Prop.\ \ref{prop:op-H}, given Prop.\ \ref{prop:vHv}]
Observe, first of all, that $ \tbf D =\tbf D_{\tbf m}$, by \eqref{eq:Dij} \& \eqref{eq:defDS}. Now, choosing $\eps>0$ sufficiently small and $k\in \bN$ sufficiently large, Prop.\ \ref{prop:vHv} implies that 
		\[ \inf_{(\tbf u, \tbf v) \in S_{N,\eps,k}} \Big(\tbf u,  \big  ( \tbf D_{\tbf v} +\beta^2(1-q)- \frac{2\beta^2}N \tbf v\otimes \tbf v-  \beta \tbf G\big) \tbf u \Big)   \geq  \widetilde c/2  \]
for some constant $\widetilde c >0$, with probability tending to one as $N\to \infty$. By the assumptions \eqref{eq:AT} \&  \eqref{eq:AT+}, we can apply Theorem \ref{thm:CT} which combined with Markov implies
		\[ \big\| \tbf m- \tbf m^{(k)} \big\|^2  \leq N \eps \]
with probability tending to one as $N\to\infty$. Thus 
		\[\begin{split}
		&\inf_{\|\tbf u\|^2=1} \big(\tbf u,  \big  ( \tbf D  +\beta^2(1-q) -  \beta \tbf G\big) \tbf u \big) \\
		 &\;\;\geq  \inf_{\|\tbf u\|^2=1} \big(\tbf u,  \big  ( \tbf D_{\tbf m}  +\beta^2(1-q) -\frac{2\beta^2}N \tbf m\otimes \tbf m-  \beta \tbf G\big) \tbf u \big)  \\
		&\;\; \geq \inf_{\substack{(\tbf u, \tbf v)\in S_{N,\eps,k}}} \Big(\tbf u,  \big  ( \tbf D_{\tbf v} +\beta^2(1-q) - \frac{2\beta^2}N \tbf v\otimes \tbf v-  \beta \tbf G\big) \tbf u \Big) \geq \widetilde c/2 , \end{split}\]
i.e. $\tbf D  +\beta^2(1-q) -  \beta \tbf G\geq c$ for $c =\widetilde c/2$, with probability tending to one as $N\to \infty$.
\end{proof}
The rest of this section is devoted to the proof of Prop.\ \ref{prop:vHv}.
\begin{proof}[Proof of Prop.\ \ref{prop:vHv}]
\noindent \tbf{Step 1:} We adapt the main idea of \cite{celentano2022sudakov} and apply the Sudakov-Fernique inequality \cite{Sud1, Fer, Sud2}, conditionally on $\cG_{k}$, to reduce the problem to a solvable variational problem. To this end, we first write 
		\[\begin{split}
		\beta\big(\tbf u,    \tbf G \tbf u \big)& =  \beta\big( \tbf u,  \tbf G^{(k+1)} \tbf u \big) + \frac{2\beta}{\sqrt N} \sum_{s=1}^{k} (\tbf u, \zeta^{(s)} \otimes \phi^{(s)}, \tbf u) +o_{\tbf u}(1)
		\end{split}\] 
for an error $o_{\tbf u}(1) = -N^{-1/2}\sum_{s=1}^{k} (\zeta^{(s)}, \phi^{(s}) | (\phi^{(s)}, \tbf u)|^2$ which satisfies $ \sup_{\|u\|=1} o_{\tbf u} (1) \to 0$ as $N\to\infty$ almost surely, by Prop.\ \ref{prop:Bolt} 5). This means that 
		\[\begin{split}
		& \text{p-}\limsup_{N \to \infty} \bigg[\inf_{(\tbf u, \tbf v) \in S_{N,\eps,k}} \Big(\tbf u,  \big  ( \beta \tbf G + \frac{2\beta^2}N \tbf v\otimes \tbf v-\tbf D_{\tbf v} -\beta^2(1-q)\big) \tbf u \Big) \bigg]\\
		&= \!\text{p-}\limsup_{N \to \infty} \!\bigg[\inf_{(\tbf u, \tbf v) \in S_{N,\eps,k}} \!\Big(\tbf u,  \big  ( \beta \tbf G^{(k+1)}\! +\!\frac{2\beta}{\sqrt N} \sum_{s=1}^{k} \! \zeta^{(s)}\! \otimes \phi^{(s)}\!+\! \frac{2\beta^2}N \tbf v\!\otimes\! \tbf v\!-\!\tbf D_{\tbf v}\! -\!\beta^2(1-q)\big) \tbf u \Big) \bigg]. 
		\end{split} \]
Now, comparing the (conditionally on $\cG_k$) Gaussian processes
		\[ \begin{split}
		\big( \tbf X_{\tbf u, \tbf v}\big)_{(\tbf u, \tbf v) \in S_{N,\eps,k}}  & = \Big(  \beta \big (\tbf u,   \tbf G^{(k+1)} \tbf u \big) + \tbf f\, (\tbf u, \tbf v) \Big)_{(\tbf u, \tbf v) \in S_{N,\eps,k}}, \\
		\big( \tbf Y_{\tbf u, \tbf v}\big)_{ (\tbf u, \tbf v) \in S_{N,\eps,k} }  & = \Big(  \frac{2 \beta}{\sqrt N}\|\tbf{Q}^{(k)} \mathbf u\|  (\tbf Q^{(k)} \tbf u , \xi ) + \tbf f\, (\tbf u, \tbf v)\Big)_{(\tbf u, \tbf v) \in S_{N,\eps,k}}, 
		\end{split}\]
where $\xi \sim \cN (0, \text{id}_{\bR^N})$ denotes a Gaussian vector independent of the remaining disorder and where we abbreviate
		\[ \tbf f\,(\tbf u, \tbf v) = \frac{2\beta}{\sqrt N} \sum_{s=1}^{k} (\tbf u, \zeta^{(s)} ) (\phi^{(s)}, \tbf u) +\! \frac{2\beta^2}N ( \tbf u, \tbf v)^2\!- (\tbf u, \tbf D_{\tbf v}\tbf u) -\beta^2(1-q),  \]
we have $ \bE_k  \tbf X_{\tbf u, \tbf v} = \bE_k \tbf Y_{\tbf u, \tbf v}  = \tbf f\,(\tbf u, \tbf v)$ and an application of Prop.\ \ref{prop:Bolt}, part 4), shows that 
		\[ \bE_k \big(  \tbf X_{\tbf u, \tbf v}-  \tbf X_{\tbf u', \tbf v'}\big)^2\leq \bE_k \big(  \tbf Y_{\tbf u, \tbf v}-  \tbf Y_{\tbf u', \tbf v'}\big)^2  \]
for every $ (\tbf u, \tbf v),(\tbf u', \tbf v')  \in S_{N,\eps,k}  $. Thus, by Vitale's extension \cite{V} of the Sudakov-Fernique inequality \cite{Sud1, Fer, Sud2} we obtain that a.s.
		\be\label{eq:SF0}\begin{split}
		\bE_k   \sup_{(\tbf u, \tbf v) \in S_{N,\eps,k}} \Big[ \beta \big (\tbf u,   \tbf G^{(k+1)} \tbf u \big) +\tbf f\,(\tbf u, \tbf v)  \Big] \leq \bE_k\!\!   \sup_{(\tbf u, \tbf v) \in S_{N,\eps,k}}\! \Big[  \frac{2 \beta}{\sqrt N}\|\tbf{Q}^{(k)} \mathbf u\|  (\tbf Q^{(k)} \tbf u , \xi )+ \tbf f\, (\tbf u, \tbf v) \Big].
		\end{split}\ee
Next, observing that conditionally on $\cG_k$, we have the distributional equality 
		\[  \tbf G^{(k+1)} \stackrel{\text{d}}{=} \tbf Q^{(k)} \frac{1}{\sqrt{2N}}\big(   \tbf U  + \tbf U^{T}\big)\tbf Q^{(k)} \] 
for some random matrix $ \tbf U = ( u_{ij})_{1\leq i,j\leq N}\in \bR^{N\times N}$ with i.i.d. entries $  u_{ij} \sim \cN(0,1)$ (without symmetry constraint) for $i\neq j$ and $ u_{ii}=0$, a standard application of Gaussian concentration (see, for instance, \cite[Theorem 1.3.4]{talagrand2010mean-vol1}) implies that 
		\[\begin{split}
		&\bP_{k} \bigg(\,  \bigg| \sup_{(\tbf u, \tbf v) \in S_{N,\eps,k}} \bigg[ \beta \big (\tbf u,   \tbf G^{(k+1)} \tbf u \big) +\tbf f\,(\tbf u, \tbf v)  \bigg] \\
		&\hspace{2cm}- \bE_k \sup_{(\tbf u, \tbf v) \in S_{N,\eps,k}} \bigg[ \beta \big (\tbf u,   \tbf G^{(k+1)} \tbf u \big) +\tbf f\,(\tbf u, \tbf v)  \bigg]\bigg| > t \bigg) \leq 2e^{-CNt^2}. 		
		\end{split}\]
Indeed, for $\tbf X, \tbf Y\in \bR^{N\times N}$, notice that
		\[\begin{split}
		&\sup_{(\tbf u, \tbf v) \in S_{N,\eps,k}}\! \bigg[  \frac\beta{\sqrt{N}} \big (\tbf{Q}^{(k)} \tbf u,    \bar{\tbf {X}}\tbf{Q}^{(k)}  \tbf u \big) \!+\tbf f\,(\tbf u, \tbf v)  \bigg] -\!\!\sup_{(\tbf u, \tbf v) \in S_{N,\eps,k}} \!\bigg[  \frac\beta{\sqrt{N}} \big (\tbf{Q}^{(k)} \tbf u,    \bar{ \tbf {Y}} \tbf{Q}^{(k)}\tbf u \big) \!+\tbf f\,(\tbf u, \tbf v)  \bigg]\\
		& = \sup_{(\tbf u, \tbf v) \in S_{N,\eps,k}} \bigg[  \frac{\sqrt{2}\beta}{\sqrt{N}} \big (\tbf{Q}^{(k)} \tbf u,    ( \tbf {X} - \tbf {Y} ) \tbf{Q}^{(k)} \tbf u \big) + \frac{\sqrt{2}\beta}{\sqrt{N}} \big (\tbf{Q}^{(k)} \tbf u,    \tbf {Y} \tbf{Q}^{(k)} \tbf u \big) +\tbf f\,(\tbf u,  \tbf v)  \bigg] \\
		&\hspace{0.5cm}-\sup_{(\tbf u, \tbf v) \in S_{N,\eps,k}} \bigg[  \frac{2\beta}{\sqrt{N}} \big (\tbf{Q}^{(k)} \tbf u,    \tbf {Y} \tbf{Q}^{(k)} \tbf u \big) +\tbf f\,(\tbf u, \tbf v)  \bigg]\\
		&\leq \sup_{\tbf u' \in \bR^N: \| \tbf u'\|^2\leq 1 }  \frac{\sqrt{2}\beta}{\sqrt{N}} \big ( \tbf u',    ( \tbf {X} - \tbf {Y} )\, \tbf u' \big) \leq \frac{C}{\sqrt{N}} \bigg(\sum_{i,j=1}^N (\tbf X- \tbf Y)_{ij}^2\bigg)^{1/2} =  \frac{C}{\sqrt{N}} \| \tbf X- \tbf Y\|
		\end{split}\]
so that, upon switching the roles of $\tbf X \in \bR^{N\times N}$ and $\tbf Y\in \bR^{N\times N}$, we get
		\[\begin{split}
		&\bigg|\sup_{(\tbf u, \tbf v) \in S_{N,\eps,k}} \!\!\bigg[  \frac\beta{\sqrt{N}} \big (\tbf{Q}^{(k)} \tbf u,    \bar{\tbf {X}}\tbf{Q}^{(k)}  \tbf u \big) \!+\tbf f\,(\tbf u, \tbf v)  \bigg] -\!\!\!\sup_{(\tbf u, \tbf v) \in S_{N,\eps,k}}\!\! \bigg[  \frac\beta{\sqrt{N}} \big (\tbf{Q}^{(k)} \tbf u,    \bar{ \tbf {Y}} \tbf{Q}^{(k)}\tbf u \big) \!+\tbf f\,(\tbf u, \tbf v)  \bigg]\bigg|\\
		&\leq \frac{C}{\sqrt{N}} \| \tbf X- \tbf Y\|
		\end{split}\]
for the deterministic constant $C = \sqrt{2} \beta >0$. Arguing along the same lines, we also find
		\[\begin{split}
		&\bP_{k} \bigg(\,  \bigg| \sup_{(\tbf u, \tbf v) \in S_{N,\eps,k}} \bigg[ \frac{2 \beta}{\sqrt N}\|\tbf{Q}^{(k)} \mathbf u\|  (\tbf{Q}^{(k)}  \tbf u , \xi ) +\tbf f\,(\tbf u, \tbf v)  \bigg] \\
		&\hspace{2cm}- \bE_k \sup_{(\tbf u, \tbf v) \in S_{N,\eps,k}} \bigg[ \frac{2 \beta}{\sqrt N}\|\tbf{Q}^{(k)} \mathbf u\|  (\tbf{Q}^{(k)} \tbf u , \xi )  +\tbf f\,(\tbf u, \tbf v)  \bigg]\bigg| > t \bigg) \leq 2 e^{-CNt^2}, 
		\end{split}\]
and since $C>0$ is independent of $\cG_k$, we can take the expectation $\bE(\cdot)$ over the previous two tail bounds to see that they hold true unconditionally. Thus, we conclude
		\[\begin{split} 
		 \lim_{N\to\infty } \bigg| \!\sup_{(\tbf u, \tbf v) \in S_{N,\eps,k}} \Big[ \beta \big (\tbf u,   \tbf G^{(k+1)} \tbf u \big) +\tbf f\,(\tbf u, \tbf v)  \Big] - \bE_k \!\!\sup_{(\tbf u, \tbf v) \in S_{N,\eps,k}} \Big[ \beta \big (\tbf u,   \tbf G^{(k+1)} \tbf u \big) +\tbf f\,(\tbf u, \tbf v)  \Big]\bigg| &= 0 \end{split}\]
and 
		\[\begin{split}
		 &\lim_{N\to\infty } \bigg| \sup_{(\tbf u, \tbf v) \in S_{N,\eps,k}} \bigg[ \frac{2 \beta}{\sqrt N}\|\tbf{Q}^{(k)} \mathbf u\|  ( \tbf{Q}^{(k)}\tbf u , \xi )  +\tbf f\,(\tbf u, \tbf v)  \bigg] \\
		 &\hspace{4cm}- \bE_k \sup_{(\tbf u, \tbf v) \in S_{N,\eps,k}} \bigg[ \frac{2 \beta}{\sqrt N}\|\tbf{Q}^{(k)} \mathbf u\|  ( \tbf{Q}^{(k)}\tbf u , \xi )  +\tbf f\,(\tbf u, \tbf v)  \bigg]\bigg| = 0
		 \end{split}\]
in the sense of probability. Combining these observations with \eqref{eq:SF0}, we finally arrive at
		\be\label{eq:SF}\begin{split}
		 &\limsup_{\eps\to 0}\limsup_{k\to\infty} \text{p-}\limsup_{N \to \infty}\bigg[  \sup_{(\tbf u, \tbf v) \in S_{N,\eps,k}} \Big(\tbf u,  \big  (  \beta \tbf G +\frac{2\beta^2}N\tbf v\otimes \tbf v- \tbf D_{\tbf v}-\beta^2(1-q) \big) \tbf u \Big) \bigg] \\
		&\leq \limsup_{\eps\to 0}\limsup_{k\to\infty} \text{p-}\limsup_{N \to \infty} \bigg[\sup_{(\tbf u, \tbf v) \in S_{N,\eps,k}}  \text{F}\big(\mathbf u,\mathbf v,\xi, \mathbf M^{(k)},\mathbf Z^{(k)}\big)\bigg] ,
		\end{split}\ee
where $ \tbf M^{(k)}$ and $ \tbf Z^{(k)}$ denote the matrices
		$$\mathbf M^{(k)} = (\mathbf m^{(1)}\,\mathbf m^{(2)}\,\ldots\,\mathbf m^{(k)})\in \bR^{N\times k},\hspace{0.5cm} \mathbf Z^{(k)} = (\mathbf z^{(2)},\ldots,\mathbf z^{(k+1)})\in \bR^{N\times k}$$
and where
		\[\begin{split}
		\text{F}\big(\mathbf u,\mathbf v,\xi, \mathbf M^{(k)},\mathbf Z^{(k)}\big) &=\frac{2 \beta}{\sqrt N}\|\tbf{Q}^{(k)} \mathbf u\|  ( \tbf u , \xi )+ 2\beta  \big( \mathbf u, \mathbf Z^{(k)}\big( \mathbf M^{(k)T}\mathbf M^{(k)}\big)^{-1}\mathbf M^{(k)T} \mathbf u\big) \\
		&\hspace{0.5cm} +\frac{2\beta^2}N(\tbf u, \tbf m^{(k)})^2 - ( \tbf u , \tbf D_{\tbf v} \tbf u ) -\beta^2(1-q). 
		\end{split}\]		
Here, we used additionally that
		\[ \sup_{\| \tbf v - \tbf m^{(k)}\|^2\leq N\eps } N^{-1} \big\|  \tbf v \otimes \tbf v - \tbf m^{(k)}\otimes  \tbf m^{(k)}\big\|_{\text{op}} \leq 2 \sqrt{\eps} \]
and that, in the sense of probability, we have
		\[\begin{split}
		\lim_{N\to\infty}\big|  \sup_{\| \tbf u\|^2=1} N^{-1/2} \big(\tbf P^{(k)}\tbf u  , \xi \big) \big|\leq \lim_{N\to\infty} N^{-1/2}  \big \|  \tbf P^{(k)}  \xi \big\|   &= 0,\\
		 \lim_{N\to\infty} \Big\| \frac{1}{\sqrt N} \sum_{s=1}^k \zeta^{(s)} \otimes \phi^{(s)} -   \,\mathbf Z^{(k)}\big( \mathbf M^{(k)T}\mathbf M^{(k)}\big)^{-1}\mathbf M^{(k)T}\Big\|_{\text{op}} &=0.
		\end{split}  \]
The last two identities readily follow e.g. from a simple second moment bound and, respectively, from Prop.\ \ref{prop:Bolt}. Notice that both matrices $ N^{-1/2} \sum_{s=1}^k \zeta^{(s)} \otimes \phi^{(s)} $ and $\mathbf Z^{(k)}\big( \mathbf M^{(k)T}\mathbf M^{(k)}\big)^{-1}\mathbf M^{(k)T}$ are at most of rank $k$ ($ \tbf Q^{(k)}\bR^N$ contained in their kernels) so that the norm convergence follows from pointwise convergence on the vectors $ (\tbf m^{(s)})_{s=1}^k$. 

%	\be\label{eq:SF}\begin{split}
%	&\limsup_{\eps\to 0}\limsup_{k\to\infty} \text{p-}\limsup_{N \to \infty} \sup_{(\tbf u, \tbf v)\in S_{N,\eps,k}} \Big(\tbf u,  \big  (  \beta \tbf G+\frac{2\beta^2}N\tbf v\otimes \tbf v- \tbf D_{\tbf v} -\beta^2(1-q) \big) \tbf u \Big) \\
%	&\hspace{3.8cm}\leq \limsup_{\eps\to 0}\limsup_{k\to\infty}\text{p-}\limsup_{N \to \infty} \sup_{(\tbf u, \tbf v)\in S_{N,\eps,k}}  \text{F}'\big(\mathbf u,\mathbf v,\xi, \mathbf M^{(k)},\mathbf Z^{(k)}\big), 
%	\end{split}\ee
%where
%		\[\begin{split}
%		 \text{F}'\big(\mathbf u,\mathbf v,\xi, \mathbf M^{(k)},\mathbf Z^{(k)}\big) &= 2\beta  \big( \mathbf u, \mathbf Z^{(k)}\big( \mathbf M^{(k)T}\mathbf M^{(k)}\big)^{-1}\mathbf M^{(k)T} \mathbf u\big) + \frac{2 \beta}{\sqrt N}\|\tbf{Q}^{(k)} \mathbf u\|  ( \tbf u , \xi )\\
%		&\hspace{0.5cm} +\frac{2\beta^2}N(\tbf m^{(k)}, \tbf u)^2 - ( \tbf u , \tbf D_{\tbf v} \tbf u ) -\beta^2(1-q). 
%		\end{split}\]

\textbf{Step 2:} The remainder of the proof is based on analyzing the optimization problem on the r.h.s. of Eq.\ \eqref{eq:SF}. This can be done with the same arguments as in \cite[Sections 4.1 to 4.4]{celentano2022sudakov}. For completeness, we translate the arguments from \cite{celentano2022sudakov} to the present setting. In particular, we point out at which steps the AT condition \eqref{eq:AT} enters the analysis. 

First, we control the r.h.s. of Eq.\ \eqref{eq:SF} through an optimization problem that only involves the limiting distributions of $ \tbf M^{(k)}$ and $\tbf Z^{(k)}$, based on Theorem \ref{thm:law}. Setting  
		\[   \mu^{(k)}_N = \frac1N \sum_{i=1}^N \delta_{\sqrt{N} u_i, v_i, \xi_i, \tbf M^{(k)}_{i\cdot}\!,\, \tbf Z^{(k)}_{i\cdot}  } \in \cP(\bR^{2k+3} ), \]
where by $ \tbf X_{i\cdot} \in \bR^{n} $ we denote the $i$-th row of $\tbf X\in \bR^{m\times n}$, and setting 
		$$ \bE_{ \mu } (f) = \int \mu (dx) \,f(x), \hspace{0.2cm} \bP_{ \mu } (S) = \int \mu (dx) \,\textbf{1}_S(x) =\mu(S), \hspace{0.2cm}   \langle f, g\rangle_{L^2( d\mu) }  = \int \mu (dx)\,  f(x)g(x)$$ 
for $\mu\in\cP( \bR^{2k+3})$, we have that
		\[  \text{F}\big(\mathbf u,\mathbf v,\xi, \mathbf M^{(k)},\mathbf Z^{(k)}\big)=\Phi \big (\mu_N^{(k)} \big), \] 
where for every $\mu\in\cP( \bR^{2k+3})$ with finite second moment, we define
		\[\begin{split}
		 \Phi(\mu) &= 2\beta \,  \bE_{ \mu } ( \tU  \ul{ \text{Z}}^T  )     \big[ \bE_{ \mu } ( \ul{ \text{M}} \, \ul{ \text{M}}^T)  \big]^{-1}   \bE_{ \mu } (\ul{ \text{M}} \tU  ) +2\beta \|\tbf Q_{\ul{ \text{M}}} \tU \|_{L^2(d\mu)} \langle \tU, \Xi \rangle_{L^2(d\mu)} \\
		 &\hspace{0.5cm} + 2\beta^2 \langle  M_k, \tU \rangle^2_{L^2(d\mu)} -  \bE_{ \mu }\big(   (1- \tV^2)^{-1} \tU^2\big) -\beta^2(1-q). 
		  \end{split}\]
Here, the coordinates in the integrals over $\bR^{2k+3}$ w.r.t. to $\mu$ are denoted by $ (\tU, \tV, \Xi, \ul{\tM}, \ul{\tZ})  $ (or, in other words, $ \cL(\tU, \tV, \Xi, \ul{\tM}, \ul{\tZ} )=\mu$) and $\tbf Q_{\ul{ \text{M}} } $ denotes the projection onto the orthogonal complement of the coordinates of $\ul{ \text{M}} = (M_1,\dots, M_k)$, i.e. 
	\[  \tbf Q_{\ul{ \text{M}} } \tU = \tU -\ul{ \tM}^T\big( \bE_{ \mu } \, \ul{ \text{M}}\,\ul{ \text{M}}^T  \big)^{-1} \bE_\mu\, \ul{ \text{M}} \tU.\]
To stay consistent with the construction in \eqref{eq:AMP} and with the previous notation, we write in the following $ \ul{\tZ}=(Z_2,\dots,Z_{k+1})$. Now, applying \cite[Lemma 1]{celentano2022sudakov} (whose proof in \cite[Appendix B.2]{celentano2022sudakov} carries over directly to the present context upon replacing the functional $ \text{F}_{\text{x},k}$ from \cite{celentano2022sudakov} by $\Phi$ defined above, based on Prop.\ \ref{prop:Bolt}), we obtain the upper bound
		\be\label{eq:Phibnd}\begin{split}
		\limsup_{\eps\to 0}\limsup_{k\to\infty}\text{p-}\limsup_{N \to \infty} \sup_{(\tbf u, \tbf v)\in S_{N,\eps,k}}  \text{F}\big(\mathbf u,\mathbf v,\xi, \mathbf M^{(k)},\mathbf Z^{(k)}\big)
		&\leq \limsup_{\eps\to 0}\limsup_{k\to\infty} \sup_{\mu \in \cS_{\eps,k} } \Phi(\mu),
		\end{split}\ee
where
		\[\begin{split}
		\cS_{\eps,k}= \big\{&\mu= \cL(\tU, \tV, \Xi, \ul{\tM}, \ul{\tZ}):\, \|\tU \|_{L^2(d\mu)}^2=1,\; \| \tV- M_k\|^2_{L^2(d\mu)}\leq \eps, \bP_\mu( |\tV|<1)=1,\\
		&\,\ul{\text{Z}} = (Z_2,\ldots, Z_{k+1})\sim \cN(\mathbf{0},\mathbf K_{\leq k}),\,\ul{\text{M}} = (M_1,\ldots, M_k) \text{ with } M_1 = \sqrt{q}, \\
		&\,M_s \stackrel{s>1}= \tanh(h + \beta Z_{s-1}),\, \Xi \sim \cN(0,1) \text{ independent of } \ul{\text{Z}}\,\big\}.
		\end{split}\]
Next, using that $ \bE_{ \mu } ( \ul{ \text{Z}} \, \ul{ \text{Z}}^T) =  \bE_{ \mu } ( \ul{ \text{M}} \, \ul{ \text{M}}^T)  = \tbf{K}_{\leq k}$, we observe that 
		\[\begin{split}
		  & \big\| \ul{ \text{Z}}^T       \big[ \bE_{ \mu } ( \ul{ \text{M}} \, \ul{ \text{M}}^T)  \big]^{-1}   \bE_{ \mu } (\ul{ \text{M}} \tU  )\big\|^2_{L^2(d\mu)} + \|\tbf Q_{\ul{ \text{M}}} \tU \|^2_{L^2(d\mu)} \| \Xi \|^2_{L^2(d\mu)}\\
		  &  =   \bE_{ \mu } (\ul{ \text{M}} \tU  )^T  \big[ \bE_{ \mu } ( \ul{ \text{M}} \, \ul{ \text{M}}^T)  \big]^{-1}   \bE_{ \mu } ( \ul{ \text{Z}} \, \ul{ \text{Z}}^T)  \big[ \bE_{ \mu } ( \ul{ \text{M}} \, \ul{ \text{M}}^T)  \big]^{-1}  \bE_{ \mu } (\ul{ \text{M}} \tU  )  + \|\tbf Q_{\ul{ \text{M}}} \tU \|_{L^2(d\mu)}^2= 1
		  \end{split}\]
so that, by the independence  of $ \ul{ \tZ}$ and $\Xi$, we have that under $\mu$
		\[ \ul{ \text{Z}}^T   \big[ \bE_{ \mu } ( \ul{ \text{M}} \, \ul{ \text{M}}^T)  \big]^{-1}   \bE_{ \mu } (\ul{ \text{M}} \tU  ) + \|\tbf Q_{\ul{ \text{M}}} \tU \|_{L^2(d\mu)}  \Xi\sim \cN(0,1).\]
Arguing similarly that 
		\[ \big\langle Z_{k+1}, \ul{ \text{Z}}^T   \big[ \bE_{ \mu } ( \ul{ \text{M}} \, \ul{ \text{M}}^T)  \big]^{-1}   \bE_{ \mu } (\ul{ \text{M}} \tU  ) \big\rangle_{L^2(d\mu)} = \langle M_k, \tU\rangle_{L^2(d\mu)}  \]
and using that $ \| Z_{k+1}\|_{L^2(d\mu)}= \| M_k\|_{L^2(d\mu)} = \sqrt{q}  $, we can thus write 
		\[\begin{split}
		\ul{ \text{Z}}^T   \big[ \bE_{ \mu } ( \ul{ \text{M}} \, \ul{ \text{M}}^T)  \big]^{-1}   \bE_{ \mu } (\ul{ \text{M}} \tU  ) + \|\tbf Q_{\ul{ \text{M}}} \tU \|_{L^2(d\mu)}  \Xi  = q^{-1}\langle M_k, \tU \rangle_{L^2(d\mu)}  Z_{k+1} + \|\tbf Q_{M_{k}} \tU \|_{L^2(d\mu)} \, \Xi'
		\end{split}    \]
for some Gaussian $\Xi'\sim \cN(0,1)$, which is independent of $Z_{k+1}$, and where $Q_{M_k}$ denotes the projection onto the orthogonal complement of $M_k$ in $L^2(d\mu)$. In particular, we have
		\[\begin{split}
		\Phi(\mu) &= \frac{2\beta}q \, \langle \tU, M_k\rangle_{L^2(d\mu)} \langle Z_{k+1}, \tU\rangle_{L^2(d\mu)} +2\beta \|\tbf Q_{M_k} \tU \|_{L^2(d\mu)} \langle \tU, \Xi'\rangle_{L^2(d\mu)} \\
		 &\hspace{0.5cm} + 2\beta^2 \langle  M_k, U \rangle^2_{L^2(d\mu)} -  \bE_{ \mu }\big(   (1- \tV^2)^{-1} \tU^2\big) -\beta^2(1-q) ,  
		\end{split}\]
and setting $ M_{k+1} = \tanh(h+\beta Z_{k+1}) $, such that under the AT condition \eqref{eq:AT} we have $ \| M_{k+1}- M_k\|^2_{L^2(d\mu)} = 2q- 2\alpha_k\to 0$ as $k\to \infty$ by Lemma \ref{lm:seqlemma} and Prop.\ \ref{prop:Bolt} 4), we find
		\be\label{eq:Psibnd} \limsup_{\eps\to 0}\limsup_{k\to\infty} \sup_{\mu \in \cS_{\eps,k} } \Phi(\mu) \leq \limsup_{\eps\to 0}  \sup_{\mu \in \cS_{\eps} } \Psi(\mu),  \ee
where $\Psi$ is defined by			
		\[\begin{split}
		\Psi(\mu) &=  2\beta q^{-1/2} \, \langle \tU, \tM \rangle_{L^2(d\mu)} \langle \tZ, \tU\rangle_{L^2(d\mu)} +2\beta \sqrt{ 1 - q^{-1} \langle \tM, \tU \rangle^2_{L^2(d\mu)}} \langle \tW, \tU\rangle_{L^2(d\mu)} \\
		 &\hspace{0.5cm} + 2\beta^2   \langle\tM, \tU \rangle^2_{L^2(d\mu)} -  \bE_{ \mu }\big(   (1- \tV^2)^{-1} \tU^2\big) -\beta^2(1-q), 
		 \end{split}\]
for every $\mu = \cL( \tM,  \tU, \tV, \tW, \tZ  ) \in \cS_{\eps}$, with $\cS_{\eps}$ defined by
		\[\begin{split}
		  \cS_{\eps}  = \big\{&\cL( \tM,  \tU, \tV, \tW, \tZ  ):   \| \tU\|_{L^2(d\mu)}^2= 1, \| \tV- \tM\|^2_{L^2(d\mu)}\leq \eps,  \\
		  &\;\;\bP_\mu(| \tV| <1)=1, (\tW, \tZ) \sim \cN\big( 0,\text{id}_{\bR^2}\big),   \tM  = \tanh(h+\beta \sqrt{q} \tZ )   \big\}.
		\end{split}\]
		
\textbf{Step 3:} Finally, we need to analyze the optimization problem on the r.h.s. in Eq.\ \eqref{eq:Psibnd}. Here, we introduce Lagrange multipliers and simply upper bound the r.h.s.  by 
		\[\begin{split}
		\Psi(\mu) &=  2\beta q^{-1/2} \, \langle \tU, \tM \rangle_{L^2(d\mu)} \langle \tZ, \tU\rangle_{L^2(d\mu)} +2\beta \sqrt{ 1 - q^{-1} \langle \tM, \tU \rangle^2_{L^2(d\mu)}} \langle \tW, \tU\rangle_{L^2(d\mu)} \\
		 &\hspace{0.5cm} + 2\beta^2   \langle\tM, \tU \rangle^2_{L^2(d\mu)} -  \bE_{ \mu }\big(   (1- \tV^2)^{-1} \tU^2\big) -\beta^2(1-q)\\
		 & = 2\beta  q^{-1/2} x \,  \, \langle \tZ, \tU\rangle_{L^2(d\mu)} +2\beta \sqrt{ 1 - q^{-1} x^2}\, \langle \tW, \tU\rangle_{L^2(d\mu)}+ 2\beta^2  x^2-\beta^2(1-q) \\
		 &\hspace{0.5cm} +\lambda_u(\|U\|_{L^2(d\mu)}^2 - 1) + \lambda_x ( \langle \tU, \tM \rangle_{L^2(d\mu)}- x)   -  \bE_{ \mu }\big(   (1- \tV^2)^{-1} \tU^2\big) \\
		 &\leq   2\beta^2  x^2-\beta^2(1-q) - \lambda_u -\lambda_x x  + \bE_\mu \Theta (\tM, \tV, \tW, \tZ  ),
		\end{split}\]
where we set $x = x(\tU,\tM)= \langle \tU, \tM \rangle_{L^2(d\mu)}$ as well as
		\[\begin{split}
		\Theta (m, v, w, z  ) = \sup_{u\in\bR} \Big[ \big( 2\beta  q^{-1/2} x z + 2\beta \sqrt{ 1 - q^{-1} x^2}\, w +\lambda_x m\big)u + \big(\lambda_u -(1-v^2)^{-1}\big) u^2    \Big].
		\end{split}\]
Fixing from now on $ \lambda_u<1$ and assuming w.l.o.g. $|v| <1$ (in accordance to the fact that $ \bP_\mu(|V|<1) = 1$), strict concavity implies that 
		\[\Theta (m, v, w, z  )  = \frac14 \frac{ \big( 2\beta  q^{-1/2} x z + 2\beta \sqrt{ 1 - q^{-1} x^2}\, w +\lambda_x m\big)^2    }{\big( (1-v^2)^{-1} -\lambda_u\big)}.\]
Together with the global Lipschitz continuity of $ [-1,1] \ni v\mapsto \big( (1-v^2)^{-1} -\lambda_u\big)^{-1} \in [ 0, \infty) $, we thus obtain for every fixed $\lambda_u<1$ and $K>0$ the simple upper bound
		\be\label{eq:lastbnd}\begin{split}
		&\limsup_{\eps\to 0} \sup_{\mu \in \cS_{\eps} } \Psi(\mu) \\
		& \leq \max_{|x| \leq 1}\min_{ | \lambda_x|  \leq K }\bigg[  2\beta^2  x^2-\beta^2(1-q) - \lambda_u -\lambda_x x \\
		&\hspace{2.5cm}+   \bE_\mu  \frac{ \big( 2\beta  q^{-1/2} x \tZ + 2\beta \sqrt{ 1 - q^{-1} x^2}\, \tW +\lambda_x \tM \big)^2    }{ 4( \cosh^2(h +\beta \sqrt q \tZ) -\lambda_u )} \bigg]\\
		& = \max_{|x| \leq 1}\min_{ | \lambda_x|  \leq K }\bigg[  2\beta^2  x^2-\beta^2(1-q) - \lambda_u   +\bE_\mu  \frac{\beta^2(1-q^{-1}x^2) }{( \cosh^2(h +\beta \sqrt q \tZ) -\lambda_u)}\\
		&\hspace{3.2cm} -\lambda_x x+    \bE_\mu  \frac{ \beta^2 q^{-1}x^2\tZ^2 + \beta q^{-1/2} x\lambda_x \tM\tZ + \lambda_x^2 \tM^2/4  }{( \cosh^2(h +\beta \sqrt q \tZ) -\lambda_u )} \bigg].
		\end{split} \ee
Notice that, in the last step, we used the independence of $\tW$ of $\tZ$ and, consequently, of $\tM = \tanh( h+\beta \sqrt{q} \tZ)$. Applying Gaussian integration by parts w.r.t. $\tZ$ and rescaling $x$ and $\lambda_x$ by a factor $1/\sqrt{q}$ and respectively $\sqrt{q}$ , we find that
		\be\label{eq:lastbnd}\begin{split}
		& \max_{|x| \leq 1}\min_{ | \lambda_x|  \leq K }\bigg[  2\beta^2  x^2-\beta^2(1-q) - \lambda_u   +\bE_\mu  \frac{\beta^2(1-q^{-1}x^2) }{( \cosh^2(h +\beta \sqrt q \tZ) -\lambda_u)}\\
		&\hspace{3cm} -\lambda_x x+    \bE_\mu  \frac{ \beta^2 q^{-1}x^2\tZ^2 + \beta q^{-1/2} x\lambda_x \tM\tZ + \lambda_x^2 \tM^2/4  }{( \cosh^2(h +\beta \sqrt q \tZ) -\lambda_u )} \bigg]\\
		&= \max_{|x| \leq q^{-1/2}}\bigg\{  2  \beta^2 q x^2-\beta^2(1-q) - \lambda_u   + \bE_\mu  \frac{\beta^2}{( \cosh^2(h +\beta \sqrt q \tZ) -\lambda_u)}\\
		&\hspace{2.2cm} -  \bE_\mu  \frac{2 \beta^4 q (\cosh^2(h +\beta \sqrt q \tZ) +\sinh^2(h +\beta \sqrt q \tZ))}{( \cosh^2(h +\beta \sqrt q \tZ) -\lambda_u)^2}x^2\\
		&\hspace{2.2cm} +\bE_\mu  \frac{8 \beta^4 q \cosh^2(h +\beta \sqrt q \tZ)\sinh^2(h +\beta \sqrt q \tZ)}{( \cosh^2(h +\beta \sqrt q \tZ) -\lambda_u)^3}x^2\\
		&\hspace{2.2cm}+ \min_{ | \lambda_x|  \leq \sqrt{q} K } \bE_\mu\bigg[       \frac{ \beta^2   \, }{\cosh^2(h +\beta \sqrt q \tZ) ( \cosh^2(h +\beta \sqrt q \tZ) -\lambda_u )}x\lambda_x\\
		&\hspace{4.8cm} - \frac{2\beta^2 \cosh^2 (h +\beta \sqrt q \tZ)\tanh^2(h +\beta \sqrt q \tZ)  }{  ( \cosh^2(h +\beta \sqrt q \tZ) -\lambda_u )^2} x\lambda_x \\
		&\hspace{6.2cm}  - x\lambda_x+  \frac14  \frac{  \tM^2  }{( \cosh^2(h +\beta \sqrt q \tZ) -\lambda_u )} \lambda_x^2 \bigg]\bigg\}  \\
		&=\max_{|x| \leq q^{-1/2}} \Sigma_{\lambda_u} (x). 
		\end{split} \ee
Now, evaluating the Hessian of $\Sigma_{\lambda_u}$ for $\lambda_u=0$, one obtains with $ q_4 = \bE_{\mu} \,\tM^4$ that
		\[\begin{split}
		\frac12 \frac{d^2 \Sigma_{ 0}}{dx^2}  & = 2\beta^2 q  - 2\beta^4 q \, (1 -4q +3q_4) -  (q-q_4)^{-1}\big(1-   \beta^2 (1-4q+3q_4) \big)^2 \\ 
		&= \frac{ (1-   \beta^2 (1-4q+3q_4) )}{(q-q_4)}  \big( 2\beta^2 q (q-q_4)   -  \big(1-   \beta^2 (1-2q+q_4) \big) -2\beta^2  (q-q_4)  \big)\\
		& =  \frac{ \big(1-   \beta^2 (1-4q+3q_4) \big)}{(q-q_4)}  \big(  -  \big(1-   \beta^2 (1-2q+q_4)\big)-2\beta^2 (1-q) (q-q_4)    \big).
		\end{split}\]
Finally, by the AT condition \eqref{eq:AT} and $q-q_4 = \bE_\mu \big( \frac{\tanh^2} {\cosh^2}\big)(h+\beta\sqrt{q}\,\tZ) \geq 0$, we get 
		\[ \begin{split}
		1-   \beta^2 (1-2q+q_4) &= 1- \beta^2\bE_\mu \text{sech}^4(h+\beta{\sqrt q}\, \tZ) > 0, \\
		1-   \beta^2 (1-4q+3q_4)& = 1-   \beta^2 (1-2q+q_4) + 2\beta^2 (q-q_4) > 0,
		\end{split} \]
and together with $ 1-q = \bE_{\mu} \text{sech}^2(h+\beta\sqrt{q}\,\tZ) \geq 0$ and smoothness in $\lambda_u$, this implies
		\[   \frac{d^2 \Sigma_{ \lambda_u}}{dx^2} <0 \]
for $ \lambda_u \in (0,1)$ sufficiently small. By concavity, we obtain for $\lambda_u>0 $ small enough that
		\[\begin{split}
		& \max_{|x| \leq 1}\min_{ | \lambda_x|  \leq K }\bigg[  2\beta^2  x^2-\beta^2(1-q) - \lambda_u   +\bE_\mu  \frac{\beta^2(1-q^{-1}x^2) }{( \cosh^2(h +\beta \sqrt q\, \tZ) -\lambda_u)}\\
		&\hspace{3cm} -\lambda_x x+    \bE_\mu  \frac{ \beta^2 q^{-1}x^2\tZ^2 + \beta q^{-1/2} x\lambda_x \tM\tZ + \lambda_x^2 \tM^2/4  }{( \cosh^2(h +\beta \sqrt q\, \tZ) -\lambda_u )} \bigg]\\
		& = -\beta^2(1-q) - \lambda_u +\beta^2\bE_\mu  \frac{ 1} { \cosh^2(h +\beta \sqrt q\, \tZ) -\lambda_u )} \\
		& = - \int_0^{\lambda_u} dt \, \bigg[ 1- \beta^2\bE_\mu \frac{1}{(\cosh^2(h +\beta \sqrt q\, \tZ) -t )^2}\bigg]\leq - c <0
		\end{split} \]
for some positive constant $c = c_{\beta,h}>0$. Here, the last inequality follows from the assumption that $\lambda_u>0$ is sufficiently small and from the AT condition \eqref{eq:AT}, noting that
		\[ \begin{split}
		%&\bigg[ - \lambda_u +\beta^2\bE_\mu  \frac{ 1} { (\cosh^2(h +\beta \sqrt q\, \tZ) -\lambda_u )} \bigg]_{| \lambda_u=0} = \beta^2(1-q), \\
		& -\bigg[  1 -\beta^2\bE_\mu  \frac{ 1} { (\cosh^2(h +\beta \sqrt q\, \tZ) -t )^2} \bigg]_{| t=0} = -\big( 1- \beta^2\bE_\mu \text{sech}^4(h +\beta \sqrt q\, \tZ)\big) <0.
		\end{split}\] 
Combining this with \eqref{eq:SF}, \eqref{eq:Phibnd}, \eqref{eq:Psibnd} and \eqref{eq:lastbnd}, this proves Prop.\ \ref{prop:vHv}.
\end{proof}

%%%%%%%%%%%%%%%%%%%%%%%%%%%%%%%%%%%%%%%%%%%%
%%%%%%%%%%%%%%%%%%%%%%%%%%%%%%%%%%%%%%%%%%%%
%%%%%%%%%%%%%%%%%%%%%%%%%%%%%%%%%%%%%%%%%%%%

\vspace{0.3cm}
\noindent\textbf{Acknowledgements.} We thank two anonymous referees for helpful comments. C. B. acknowledges support by the Deutsche Forschungsgemeinschaft (DFG, German Research Foundation) under Germany’s Excellence Strategy – GZ 2047/1, Projekt-ID 390685813. The work of C. X. is partially funded by a Simons Investigator award. The work of H.-T. Y. is partially supported by the NSF grant DMS-1855509 and DMS-2153335 and a Simons Investigator award.

\vspace{0.8cm}
\noindent {\Large \textbf{Data Availability}}
\vspace{0.5cm}

\noindent Data sharing is not applicable to this article as no new data were created or analyzed in this study.

%\vspace{0.5cm}
%\noindent {\Large \textbf{Conflict of Interests}}
%\vspace{0.5cm}

%\noindent The authors have no conflict of interest to declare that are relevant to the content of this article.
%%%%%%%%%%%%%%%%%%%%%%%%%%%%%%%%%%%%%%%%%%%%
%%%%%%%%%%%%%%%%%%%%%%%%%%%%%%%%%%%%%%%%%%%%
%%%%%%%%%%%%%%%%%%%%%%%%%%%%%%%%%%%%%%%%%%%%

\small

\bibliography{ref}

\bibliographystyle{abbrv}

\end{document}